\pdfoutput=1
\RequirePackage{ifpdf}
\ifpdf % We~are running pdfTeX in pdf mode
\documentclass[pdftex]{sigma}
\else
\documentclass{sigma}
\fi

\numberwithin{equation}{section}

\newtheorem{Theorem}{Theorem}[section]
\newtheorem{Corollary}[Theorem]{Corollary}
\newtheorem{Lemma}[Theorem]{Lemma}
\newtheorem{Proposition}[Theorem]{Proposition}

 { \theoremstyle{definition}
\newtheorem{Definition}[Theorem]{Definition}
\newtheorem{Remark}[Theorem]{Remark}

 }

\usepackage{import}

\usepackage{braket}
\usepackage{bm}

\usepackage{tikz}

\usepackage{subfig}
\usepackage[labelsep=period,labelfont=bf,font=small]{caption}

\let\Re\undefined

\let\Im\undefined

\DeclareMathOperator{\Re}{Re}

\DeclareMathOperator{\Im}{Im}

\definecolor{shadecolor}{rgb}{0.95, 0.95, 0.86}

\usetikzlibrary{decorations.markings}

\def\bt{\begin{Theorem}}
	
\def\et{\end{Theorem}}

\def\be{\begin{equation}}
	
	\def\ee{\end{equation}}

\def\bi{\begin{itemize}}
	
	\def\ei{\end{itemize}}

\def\bea{\begin{eqnarray}}
	
	\def\eea{\end{eqnarray}}

\def\bl{\begin{Lemma}}
	
	\def\el{\end{Lemma}}

\def\bd{\begin{Definition}}
	
	\def\ed{\end{Definition}}

\def\bp{\begin{Proposition}}
	
	\def\ep{\end{Proposition}}

\def\br{\begin{Remark}}
	
	\def\er{\end{Remark}}

\def\bc{\begin{Corollary}}
	
	\def\ec{\end{Corollary}}

\def\ra{\rightarrow}

\newcommand{\G}{\Gamma}

\renewcommand{\d}{\delta}

\renewcommand{\o}{\omega}

\newcommand{\g}{\gamma}

\renewcommand{\part}{\partial}

\newcommand{\Res}{\operatorname{Res}}

\def\Rscr{\mathcal{R}}

\newcommand{\bigo}{\mathcal{O}}

\newcommand{\sech}{\hbox{sech}}

\def\C{{\mathbb C}}

\def\R{{\mathbb R}}

\def\N{{\mathbb N}}

\def\a{\alpha}

\def\g{\gamma}

\def\l{\lambda}

\def\1{{\bf 1}}

\def\r{\rho}

\def\s{ {\sigma}}

\def\th{ {\theta}}

\def\x{\xi}

\def\hf{\frac{1}{2}}

\def\dn{\mathop{\rm dn}\nolimits}

\def\sech{\mathop{\rm sech}\nolimits}

\def\l{\lambda}
\def\s{\sigma}
\def\o{\omega}
\def\g{\gamma}
\def\G{\Gamma}
\def\1{{\bf 1}}

\def\d{\delta}
\def\r{\rho}
\def\th{\theta}

\def\tk{\tilde k}

\def\to{\tilde \o}

\begin{document}

\renewcommand{\thefootnote}{}

\newcommand{\arXivNumber}{2312.16406}

\renewcommand{\PaperNumber}{070}

\FirstPageHeading

\ShortArticleName{Soliton Condensates for the Focusing Nonlinear Schr\"odinger Equation}

\ArticleName{Soliton Condensates for the Focusing Nonlinear\\ Schr\"odinger Equation: a Non-Bound State Case\footnote{This paper is a~contribution to the Special Issue on Evolution Equations, Exactly Solvable Models and Random Matrices in honor of Alexander Its' 70th birthday. The~full collection is available at \href{https://www.emis.de/journals/SIGMA/Its.html}{https://www.emis.de/journals/SIGMA/Its.html}}}

\Author{Alexander TOVBIS~$^{\rm a}$ and Fudong WANG~$^{\rm bc}$}

\AuthorNameForHeading{A.~Tovbis and F.~Wang}

\Address{$^{\rm a)}$~University of Central Florida, Orlando FL, USA}
\EmailD{\href{mailto:alexander.tovbis@ucf.edu}{alexander.tovbis@ucf.edu}}

\Address{$^{\rm b)}$~School of Sciences, Great Bay University, Dongguan, P.R.~China}
\Address{$^{\rm c)}$~Great Bay Institute for Advanced Study, Dongguan, P.R.~China}
\EmailD{\href{mailto:fudong@gbu.edu.cn}{fudong@gbu.edu.cn}}

\ArticleDates{Received December 30, 2023, in final form July 16, 2024; Published online July 31, 2024}

\Abstract{In this paper, we study the spectral theory of soliton condensates -- a special limit of soliton gases -- for the focusing NLS (fNLS). In particular, we analyze the kinetic equation for the fNLS circular condensate, which represents the first example of an explicitly solvable fNLS condensate with nontrivial large scale space-time dynamics. Solution of the kinetic equation was obtained by reducing it to Whitham type equations for the endpoints of spectral arcs. We also study the rarefaction and dispersive shock waves for circular condensates, as well as calculate the corresponding average conserved quantities and the kurtosis. We want to note that one of the main objects of the spectral theory~-- the nonlinear dispersion relations~-- is introduced in the paper as some special large genus (thermodynamic) limit the Riemann bilinear identities that involve the quasimomentum and the quasienergy meromorphic differentials.}


\Classification{37K40; 35P30; 37K10}

\renewcommand{\thefootnote}{\arabic{footnote}}
\setcounter{footnote}{0}

\section{Introduction}

\subsection{Soliton gases for integrable equations}
Many nonlinear integrable equations have special solutions representing solitary waves~--
%a single hump
localized traveling wave solutions also known as solitons. Solitons have some very peculiar and well studied properties. One of the most celebrated of them is the property that two solitons with different velocities retain their shapes and velocities after the interaction, so that the only result of this interaction is a phase change (e.g., a shift in the position of the center) of the above solitons. There are also more complicated $n$-soliton solutions which can be considered as ensembles of~$n$ interacting solitons.

The spectral theory of soliton gases (see, for example, \cite{ET2020} or \cite[Section~3]{ElRew1}) is based on the natural idea to interpret solitons as particles of some gas with elastic two-particle interactions. Assuming that a given soliton (a particle) undergoes consistent interactions (phase shifts) with other solitons (particles) in the gas, it is clear that its actual velocity will be affected by the persistent interactions, namely, by the density and the characteristics of other solitons in the gas interacting with the given one. This idea goes back to the paper of V.E.~Zakharov \cite{Za71}, where the formula for the effective velocity of a soliton in a rarefied Korteweg--de Vries (KdV) soliton gas was first presented. However, the case of a dense soliton gas required a different approach, which was suggested by G.~El in \cite{El2003} for KdV soliton gases and later by G.~El and A.~Tovbis in~\cite{ET2020} for the focusing nonlinear Schr\"odinger equation (fNLS)
soliton gases. This approach is based on considering a certain large $n$ limit of special $n$-phase nonlinear wave (finite gap) solutions to an integrable system, called the thermodynamic limit. The subject of our interest in this paper will be not the large $n$ limit of finite gap solutions to integrable equations, but rather the scaled continuum limit of wavenumbers $k_j$ and of frequencies $\o_j$, associated with these solutions. These limits will be interpreted as the density of states (DOS)~$u(z)$ and the density of fluxes (DOF)~$v(z)$ respectively of the corresponding soliton gas, where $z$ denotes the spectral parameter. The main results of the paper will be described in Section~\ref{sec-Main} below, while some brief background information on spectral theory of soliton gases and soliton condensates will be given in Section~\ref{sec-KdV}.
	More detailed information on the subject can be found in the review paper \cite{ElRew1}. We now turn to a brief discussion of the thermodynamic limit -- the key concept of the spectral theory of soliton gases.

Since finite gap solutions can be represented in terms of the Riemann theta functions, it is convenient for us to describe soliton gases starting from the corresponding Riemann surfaces. In particular, genus $n$
finite gap solutions to the fNLS, considered in this paper, can be defined in terms of a genus $n$ Schwarz symmetrical hyperelliptic Riemann surface $\Rscr_n$. Additional information in the form of $n$ initial phases is required to define a particular finite gap solution. However, the
wavenumbers $k_j$ and frequencies $\o_j$ of any finite gap solution to the fNLS, associated with $\Rscr_n$, are defined in terms of $\Rscr_n$ only. In particular, let ${\rm d}p_n$, ${\rm d}q_n$ be the second kind real normalized meromorphic differentials on $\Rscr_n$ with poles only at infinity (both sheets) and the principal parts $\pm 1$ for ${\rm d}p_n$ and $\pm 2z$ for ${\rm d}q_n$ there respectively, where the spectral parameter $z\in \Rscr_n$. Here real normalized mean that all the periods of ${\rm d}p_n$, ${\rm d}q_n$ on $\Rscr_n$ are real.
Note that the above conditions uniquely define ${\rm d}p_n$, ${\rm d}q_n$ on $\Rscr_n$.
Then the vectors $\vec k$, $\vec \o$ of the (real) periods of ${\rm d}p_n$, ${\rm d}q_n$ with respect to a fixed homology basis ($\bf A$ and $\bf B$ cycles) of $\Rscr_n$ are vectors of the wavenumbers and frequencies of finite gap solutions on $\Rscr_n$, respectively, i.e.,
\begin{equation*}%\label{waven-freq}
k_j=\oint_{\bf A_j}{\rm d}p_n,\qquad \tk_j= \oint_{\bf B_j}{\rm d}p_n,\qquad
\o_j=\oint_{\bf A_j}{\rm d}q_n,\qquad \to _j= \oint_{\bf B_j}{\rm d}q_n,
\end{equation*}
where $\vec k=(k_1,\dots,k_n,\tk_1,\dots,\tk_n),$ $\vec \o=(\o_1,\dots,\o_n,\to_1,\dots,\to_n)$.
The differentials ${\rm d}p_n$, ${\rm d}q_n$ are known as quasimomentum and quasienergy differentials, respectively, \cite{FFM,ForLee}.

Denote by $w_{j,n}=w_j$ the $j$-th normalized holomorphic differential on $\Rscr_n$, $j=1,\dots,n$, that are defined by the condition $\int_{A_k}w_j =\d_{k,j}$, $k=1,\dots,n$, where $\d_{k,j}$ is the Kronecker delta. The well known Riemann bilinear relations (RBR)
\begin{gather}
\sum_j\biggl[\oint_{\bf A_j}w_m\oint_{\bf B_j}{\rm d}p_n-\oint_{\bf A_j}{\rm d}p_n\oint_{\bf B_j}w_m\biggr]=2\pi {\rm i} {\sum} {\rm Res}\left(\int w_m {\rm d}p_n\right), \label{RBR1}\\
\sum_j\biggl[\oint_{\bf A_j}w_m\oint_{\bf B_j}{\rm d}q_n-\oint_{\bf A_j}{\rm d}q_n\oint_{\bf B_j}w_m\biggr]=2\pi {\rm i} {\sum} {\rm Res}\left(\int w_m {\rm d}q_n\right),
\label{RBR2}
\end{gather}
$m=1,\dots,n$, form systems of linear equations for $\vec k$, $\vec \o$ respectively, where the summation is taking over the only poles $z=\infty_\pm$ (infinities on both sheets of $\Rscr_n$) of ${\rm d}p_n,{\rm d}q_n$. Indeed,
taking real and imaginary parts of \eqref{RBR1}, one gets
 \begin{gather}
\sum_j k_j \Im \oint_{\bf B_j}w_m=-2\pi \Re\left( \sum {\rm Res}\left(\int w_m {\rm d}p_n\right) \right), \label{RBR1-Im} \\
\tk_m - \sum_j k_j \Re \oint_{\bf B_j}w_m=-2\pi \Im\left( \sum {\rm Res}\left(\int w_m {\rm d}p_n\right) \right),\qquad m=1,\dots,n.
\label{RBR1-Re}
\end{gather}

The matrix of the $n\times n$ system of linear equations is positive definite, since it is the imaginary part $\Im \tau$ of the Riemann period matrix \smash{$\tau= \oint_{\bf B_j}w_m$}. Once $k_j$ are known, the values of $\tk_j$ can be calculated from~\eqref{RBR1-Re}.
Thus, the systems \eqref{RBR1-Im}--\eqref{RBR1-Re} always have a unique solution.
Similar results are true for \eqref{RBR2}.
So, the wavenumbers and frequencies vectors $\vec k$, $\vec \o$ are connected via the underlying Riemann surface $\Rscr_n$. Thus, the Riemann bilinear relations for the quasimomentum and quasienergy differentials ${\rm d}p_n$, ${\rm d}q_n$, connecting the wavenumbers and frequencies, form the nonlinear dispersion relations (NDR) for the finite gap solutions to the fNLS, defined by~$\Rscr_n$.\looseness=1

One of the main subjects of the spectral theory of soliton gases is the thermodynamic limit of scaled vectors $\vec k$, $\vec\o$, i.e., the thermodynamic limit of Riemann bilinear relations \eqref{RBR1}--\eqref{RBR2}, which leads to the continuum version of the NDR established in \cite{ET2020} in the form of two Fredholm integral equations \eqref{NDR}--\eqref{NDR2}.

For simplicity of the exposition, let us assume that $N$ is even and only one (Schwarz symmetrical) band $\g_0$ intersects $\R$.
Before considering the thermodynamic limit, we want to address the choice of the homology basis. Namely, the $\bf A$ cycles are chosen as counterclockwise loops around each of the bands in $\C^+$ (on the main sheet) and clockwise loops around each of the bands in~$\C^-$ (on the main sheet). Each such loop contains only one band (branch cut) inside.
That leaves one exceptional band $\g_0$ not surrounded by a loop.

A $\bf B$ cycle $\bf B_j$ is represented by a loop going through the branch cut encircled by $\bf A_j$ and through
the exceptional branch cut $\g_0$,
see Figure \ref{Fig:Cont}.
If one starts shrinking just one branch cut of $\Rscr_n$, corresponding to $\bf A_j$ (and, generically, its Schwarz symmetrical) to a point while leaving the other branch cuts unchanged, the system \eqref{RBR1}--\eqref{RBR2} implies that the corresponding $k_j,\o_j\ra 0$, which can be associated with a solitonic limit. Indeed, an fNLS soliton is spectrally represented by a pair of Schwarz symmetrical points and a complex norming constant, the latter can be viewed as an analog of initial phases. For example, an elliptic solution to the fNLS $\psi_{m} = {\rm e}^{{\rm i}t(2-m)}\dn(x,m) $, where $\dn$ denotes the corresponding Jacobi elliptic function, in the limit $m\ra 1$ goes into a soliton solution $\psi_{1} = {\rm e}^{{\rm i}t}\sech(x) $. That observation illustrates the connection between finite gap and multi-soliton solutions.

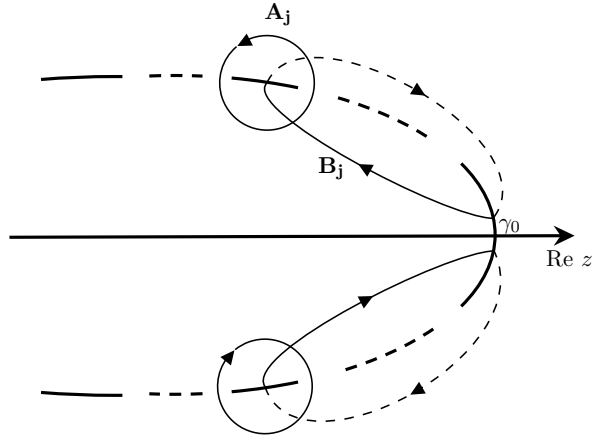
\begin{figure}[h]
	\centering
		 \resizebox{0.5\textwidth}{!}{
		
\tikzset{every picture/.style={line width=0.75pt}} %set default line width to 0.75pt

\begin{tikzpicture}[x=0.75pt,y=0.75pt,yscale=-1,xscale=1]
	%uncomment if require: \path (0,404); %set diagram left start at 0, and has height of 404
	
	%Straight Lines [id:da7223098994070081]
	\draw [fill={rgb, 255:red, 155; green, 155; blue, 155 }  ,fill opacity=1 ][line width=1.5]    (165,201) -- (519.45,200.01) ;
	\draw [shift={(523.45,200)}, rotate = 179.84] [fill={rgb, 255:red, 0; green, 0; blue, 0 }  ][line width=0.08]  [draw opacity=0] (13.4,-6.43) -- (0,0) -- (13.4,6.44) -- (8.9,0) -- cycle    ;
	%Shape: Arc [id:dp7197767916689966]
	\draw  [draw opacity=0][dash pattern={on 5.63pt off 4.5pt}][line width=1.5]  (287.44,300.29) .. controls (281.58,300.54) and (275.65,300.67) .. (269.67,300.67) .. controls (264.29,300.67) and (258.96,300.56) .. (253.69,300.36) -- (269.67,199.5) -- cycle ; \draw  [dash pattern={on 5.63pt off 4.5pt}][line width=1.5]  (287.44,300.29) .. controls (281.58,300.54) and (275.65,300.67) .. (269.67,300.67) .. controls (264.29,300.67) and (258.96,300.56) .. (253.69,300.36) ;
	%Shape: Arc [id:dp3653236378506848]
	\draw  [draw opacity=0][line width=1.5]  (451.49,154.17) .. controls (465.25,167.8) and (473,183.2) .. (473,199.5) .. controls (473,215.8) and (465.25,231.2) .. (451.49,244.83) -- (269.67,199.5) -- cycle ; \draw  [line width=1.5]  (451.49,154.17) .. controls (465.25,167.8) and (473,183.2) .. (473,199.5) .. controls (473,215.8) and (465.25,231.2) .. (451.49,244.83) ;
	%Shape: Arc [id:dp9141977916603252]
	\draw  [draw opacity=0][dash pattern={on 5.63pt off 4.5pt}][line width=1.5]  (253.69,98.64) .. controls (258.96,98.44) and (264.29,98.33) .. (269.67,98.33) .. controls (275.65,98.33) and (281.58,98.46) .. (287.44,98.71) -- (269.67,199.5) -- cycle ; \draw  [dash pattern={on 5.63pt off 4.5pt}][line width=1.5]  (253.69,98.64) .. controls (258.96,98.44) and (264.29,98.33) .. (269.67,98.33) .. controls (275.65,98.33) and (281.58,98.46) .. (287.44,98.71) ;
	%Shape: Arc [id:dp018363651430338468]
	\draw  [draw opacity=0][line width=1.5]  (305.91,99.94) .. controls (320.55,101.25) and (334.65,103.34) .. (348.02,106.12) -- (269.67,199.5) -- cycle ; \draw  [line width=1.5]  (305.91,99.94) .. controls (320.55,101.25) and (334.65,103.34) .. (348.02,106.12) ;
	%Shape: Arc [id:dp7598335590527734]
	\draw  [draw opacity=0][line width=1.5]  (348.02,292.88) .. controls (334.65,295.66) and (320.55,297.75) .. (305.91,299.06) -- (269.67,199.5) -- cycle ; \draw  [line width=1.5]  (348.02,292.88) .. controls (334.65,295.66) and (320.55,297.75) .. (305.91,299.06) ;
	%Shape: Arc [id:dp5875638612610647]
	\draw  [draw opacity=0][dash pattern={on 5.63pt off 4.5pt}][line width=1.5]  (373.39,112.47) .. controls (397.04,119.46) and (417.59,128.77) .. (433.79,139.77) -- (269.67,199.5) -- cycle ; \draw  [dash pattern={on 5.63pt off 4.5pt}][line width=1.5]  (373.39,112.47) .. controls (397.04,119.46) and (417.59,128.77) .. (433.79,139.77) ;
	%Shape: Arc [id:dp17650732443428496]
	\draw  [draw opacity=0][dash pattern={on 5.63pt off 4.5pt}][line width=1.5]  (433.79,259.23) .. controls (417.59,270.23) and (397.04,279.54) .. (373.39,286.53) -- (269.67,199.5) -- cycle ; \draw  [dash pattern={on 5.63pt off 4.5pt}][line width=1.5]  (433.79,259.23) .. controls (417.59,270.23) and (397.04,279.54) .. (373.39,286.53) ;
	%Curve Lines [id:da482825425220228]
	\draw    (328.33,103) .. controls (311,120.33) and (441,188.33) .. (471.67,189) ;
	\draw [shift={(386.13,154.55)}, rotate = 28.52] [fill={rgb, 255:red, 0; green, 0; blue, 0 }  ][line width=0.08]  [draw opacity=0] (8.93,-4.29) -- (0,0) -- (8.93,4.29) -- cycle    ;
	%Curve Lines [id:da15864062625404984]
	\draw    (327,295.67) .. controls (313,274.33) and (442.33,210.33) .. (472.33,209.67) ;
	\draw [shift={(395.71,238.3)}, rotate = 152.67] [fill={rgb, 255:red, 0; green, 0; blue, 0 }  ][line width=0.08]  [draw opacity=0] (8.93,-4.29) -- (0,0) -- (8.93,4.29) -- cycle    ;
	%Shape: Arc [id:dp5867404631491058]
	\draw  [draw opacity=0][line width=1.5]  (185.02,100.85) .. controls (200.01,99.63) and (215.49,99) .. (231.33,99) .. controls (233.1,99) and (234.87,99.01) .. (236.63,99.02) -- (231.33,200.17) -- cycle ; \draw  [line width=1.5]  (185.02,100.85) .. controls (200.01,99.63) and (215.49,99) .. (231.33,99) .. controls (233.1,99) and (234.87,99.01) .. (236.63,99.02) ;
	%Shape: Arc [id:dp19861604303444103]
	\draw  [draw opacity=0][line width=1.5]  (236.63,301.31) .. controls (234.87,301.33) and (233.1,301.33) .. (231.33,301.33) .. controls (215.49,301.33) and (200.01,300.7) .. (185.02,299.49) -- (231.33,200.17) -- cycle ; \draw  [line width=1.5]  (236.63,301.31) .. controls (234.87,301.33) and (233.1,301.33) .. (231.33,301.33) .. controls (215.49,301.33) and (200.01,300.7) .. (185.02,299.49) ;
	%Curve Lines [id:da2194430446529394]
	\draw  [dash pattern={on 4.5pt off 4.5pt}]  (328.33,103) .. controls (351.67,50.33) and (510.33,139.67) .. (471.67,189) ;
	\draw [shift={(428.85,111.13)}, rotate = 211.12] [fill={rgb, 255:red, 0; green, 0; blue, 0 }  ][line width=0.08]  [draw opacity=0] (8.93,-4.29) -- (0,0) -- (8.93,4.29) -- cycle    ;
	%Curve Lines [id:da34362702277658097]
	\draw  [dash pattern={on 4.5pt off 4.5pt}]  (327,295.67) .. controls (347.67,357.67) and (503.67,276.33) .. (472.33,209.67) ;
	\draw [shift={(419.17,300.92)}, rotate = 330.04] [fill={rgb, 255:red, 0; green, 0; blue, 0 }  ][line width=0.08]  [draw opacity=0] (8.93,-4.29) -- (0,0) -- (8.93,4.29) -- cycle    ;
	%Shape: Arc [id:dp11287113090264955]
	\draw  [draw opacity=0] (308.15,272.32) .. controls (310.8,270.18) and (313.86,268.46) .. (317.27,267.29) .. controls (332.95,261.92) and (350.01,270.27) .. (355.38,285.94) .. controls (360.75,301.61) and (352.4,318.67) .. (336.73,324.05) .. controls (321.05,329.42) and (303.99,321.07) .. (298.62,305.39) .. controls (294.49,293.33) and (298.48,280.45) .. (307.68,272.7) -- (327,295.67) -- cycle ; \draw    (308.15,272.32) .. controls (310.8,270.18) and (313.86,268.46) .. (317.27,267.29) .. controls (332.95,261.92) and (350.01,270.27) .. (355.38,285.94) .. controls (360.75,301.61) and (352.4,318.67) .. (336.73,324.05) .. controls (321.05,329.42) and (303.99,321.07) .. (298.62,305.39) .. controls (294.82,294.3) and (297.89,282.5) .. (305.58,274.66) ; \draw [shift={(307.68,272.7)}, rotate = 128.59] [fill={rgb, 255:red, 0; green, 0; blue, 0 }  ][line width=0.08]  [draw opacity=0] (8.93,-4.29) -- (0,0) -- (8.93,4.29) -- cycle    ;
	%Shape: Arc [id:dp9138993892223812]
	\draw  [draw opacity=0] (309.48,79.65) .. controls (312.13,77.51) and (315.2,75.79) .. (318.61,74.62) .. controls (334.28,69.25) and (351.34,77.6) .. (356.71,93.27) .. controls (362.08,108.95) and (353.73,126.01) .. (338.06,131.38) .. controls (322.39,136.75) and (305.33,128.4) .. (299.95,112.73) .. controls (295.82,100.67) and (299.81,87.78) .. (309.02,80.04) -- (328.33,103) -- cycle ; \draw    (311.93,77.87) .. controls (313.97,76.54) and (316.2,75.44) .. (318.61,74.62) .. controls (334.28,69.25) and (351.34,77.6) .. (356.71,93.27) .. controls (362.08,108.95) and (353.73,126.01) .. (338.06,131.38) .. controls (322.39,136.75) and (305.33,128.4) .. (299.95,112.73) .. controls (295.82,100.67) and (299.81,87.78) .. (309.02,80.04) ;  \draw [shift={(309.48,79.65)}, rotate = 332.61] [fill={rgb, 255:red, 0; green, 0; blue, 0 }  ][line width=0.08]  [draw opacity=0] (8.93,-4.29) -- (0,0) -- (8.93,4.29) -- cycle    ;
	
	% Text Node
	\draw (503.89,208) node [anchor=north west][inner sep=0.75pt]   [align=left] {$\mathrm{Re}\,\, z$};
	
	\draw (473.89,188) node [anchor=north west][inner sep=0.75pt]   [align=left] {$\gamma_0$};
	% Text Node
	\draw (359.33,148.4) node [anchor=north west][inner sep=0.75pt]    {$\mathbf{B_{j}}$};
	% Text Node
	\draw (325.33,50.73) node [anchor=north west][inner sep=0.75pt]    {$\mathbf{A_{j}}$};

\end{tikzpicture}}

	\caption{The hyperelliptic Riemann surface $\mathcal R_n$, indicating the orientations of the canonical homology basis.  }
	\label{Fig:Cont}
	\vspace{5pt}
\end{figure} 

Considering the thermodynamic limit, we assume
(see \cite{ElRew1,ET2020}) that the number $n+1$ of the branch cuts of $\Rscr_n$ is growing so that the centers of branch cut (spectral band) are accumulating on some (Schwarz symmetrical) compact $\G$ in $\C$ with a
normalized continuous density function $\phi(z)>0$ on $\G^+=\G\cap\C^+$ \big(i.e., $\int_{\G^+}\phi(w){\rm d}\l(w)=1$ for some reference measure~$\l(w)$, see Section~\ref{sec-KdV} below\big).
Simultaneously, all the bandwidth are shrinking at the order ${\rm e}^{-\nu(z)n}$, where~$\nu(z)$ is a continuous non-negative function on $\G^+$,
in such a way that the distance between any two bands should be of the order at least $\mathcal{O}(1/n)$. The wavenumbers $k_j$ and frequencies~$\o_j$ are called solitonic wavenumbers and frequencies respectively, because in the thermodynamic limit they go to zero. The function $\s(z)=\frac{ 2\nu(z)}{\phi(z)}$ is called spectral scaling function.
	In
	the thermodynamic limit, the system of linear equations \eqref{RBR1-Im} for the solitonic wavenumbers $k_j$ turns into the integral equation \eqref{NDR} for the scaled continuum limit~$u(z)$ of~$k_j$, see details in~\cite{ET2020}.
	Similarly, the imaginary part of~\eqref{RBR2} for the solitonic frequencies $\o_j$ turns into the integral equation \eqref{NDR2} for the scaled continuum limit~$v(z)$ of~$\o_j$.
	Thus, the NDR
	\eqref{NDR}--\eqref{NDR2}, discussed in Section~\ref{sec-KdV} below, which is one of the main subjects of this paper, represent the thermodynamic limit of the systems of linear equations for $k_j$, $\o_j$, i.e., the thermodynamic limit of imaginary parts of the RBR \eqref{RBR1}--\eqref{RBR2}.
So far, the derivation of the NDR
\eqref{NDR}--\eqref{NDR2} was made rigorous~\cite{TW}
in the case when $\G^+$ is a curve and $\s>0$ on $\G^+$. In the case $\s(z)\equiv 0$ on $\G^+$, i.e., in the case of sub-exponential decay of the bandwidth, a soliton gas was called soliton condensate~\cite{ET2020}. Soliton condensates have a natural extremal property: if $\G^+$ is fixed but $\s\geq 0,~\s\in C(\G^+)$, is allowed to vary, the maximal average intensity of the fNLS soliton gas is attained at $\s\equiv 0$, i.e., for the soliton condensate, see~\cite{KT2021}.\looseness=1

\subsection{fNLS bound state soliton condensate}\label{sec-KdV}

In this paper we consider soliton gases for the fNLS
\begin{equation*} %\label{NLS}
{\rm i} \psi_t + \psi_{xx} +2 |\psi|^2 \psi=0,
\end{equation*}
where $x,t\in \R$ are the space-time variables and $\psi\colon \R^2 \ra \C$ is the unknown complex-valued function.

The nonlinear dispersion relations (NDR) for the fNLS \textit{soliton gas} are
integral equations (see, for example, \cite{ET2020,TW}):
%that determine the focusing nonlinear Schr\"odinger equation (fNLS) \textbf{soliton gas} are:
\begin{align}\label{NDR}
\int_{\G^+}\log\biggl|\frac{z-\bar w}{z-w}\biggr|u(w)\,{\rm d}\l(w)+\sigma(z)u(z)&=\Im{z},\qquad z\in \G^+,\\
\int_{\G^+}\log\biggl|\frac{z-\bar w}{z-w}\biggr|v(w)|\,{\rm d}\l(w)+\sigma(z)v(z)&=-4\Re{z}\Im{z},\qquad z\in \G^+,
\label{NDR2}
\end{align}
where $\G^+\subset \overline{\C^+}$ is a compact, $\l(w)$ is a reference measure (for example, the arclength if $\G^+$ is a curve) and $\sigma(z)\colon \G^+\mapsto [0,\infty)$ is a continuous function.
As it was mentioned above, the uniqueness and existence of solutions to the discrete NDR (\eqref{RBR1-Im} and its analog for ${\rm d}q_n$) follow from the positive-definiteness of the imaginary part of the Riemann period matrix.
Similar results for the continuous NDR \eqref{NDR}--\eqref{NDR2} were obtained in \cite{KT2021},
where each of the equations
was considered as a variational equation for the minimizer of the Green's energy functional.

\begin{Definition} \label{thick}
Let $S$ be a subset of $\mathbb C$ and let $z_0 \in
\mathbb C$. Then $S$ is thick (or non-thin) at
$z_0$ if $z_0 \in \overline{S \setminus \{z_0\}}$
and if, for every superharmonic function
$u$ defined on a neighborhood of $z_0$,
\[ \liminf_{z \ra z_0 \atop z \in S \setminus \{ z_0\}} u(z) = u(z_0). \]
Otherwise, $S$ is thin at $z_0$.
\end{Definition}

A connected set with more than one point (for example, a contour) is thick at
all of its points. On the other hand, a countable set
is thin at every point. We consider $\G^+$ to be a finite collection of arc or closed regions in $\C$. Here is one of the results from \cite{KT2021}.

\begin{Theorem} \label{secondthm}
Let $\sigma$ be continuous on $\Gamma^+$, and
$S_0 = \bigl\{ z \in \Gamma^+ \mid \sigma(z) = 0 \bigr\}$.
Suppose $S_0$ is either empty or thick at each $z_0 \in S_0$
$($see Definition $\ref{thick})$. 	Then
the solution $u(z)$ to
\eqref{NDR} exists and is unique; moreover, $u(z)\geq 0$ on $\G^+$.
\end{Theorem}
Similar results can be proven for
%the second equation
\eqref{NDR2} with the exception
that, in general, $v(z)$ can take both positive and negative values on $\G^+$.
In particular, the statement of Theorem \ref{secondthm} is valid in the case $\s\equiv 0$, i.e., in the case of a soliton condensate. The ratio
\begin{gather}\label{v-eff}
s(z)=\frac{v(z)}{u(z)}
\end{gather}
represents the effective velocity of a tracer soliton (elements of the gas with the spectral characteristic $z$) in the soliton gas.

Soliton gases (and condensates) discussed so far were equilibrium soliton gases, that is, the spectral properties of these gases were uniform all over the physical $ (x,t)$-plane. One can also consider non-equilibrium soliton gases, where the parameters of the NDR \eqref{NDR}--\eqref{NDR2} are slowly varying with $ (x,t)$ on large $ (x,t)$ scales. That is, we assume that $\s=\s(z;x,t)$, $\G^+=\G^+(x,t)$ and so, the DOS $u=u(z;x,t)$ and DOF $v=v(z;x,t)$. In this case, the NDR \eqref{NDR}--\eqref{NDR2} should be supplemented by the ``continuity equation''
\begin{gather}\label{PDE}
\partial_t u(z;x,t)+\partial_x v(z;x,t)=0,
%\label{kin}
\end{gather}
thus forming the system known as kinetic equation see \cite{ET2020} (to be precise, often the NDR \eqref{NDR}--\eqref{NDR2} in the kinetic equation are replaced by the equation of state, which represents an integral equation for the unknown $s(z)$ in terms of $\G^+$ and $u(z)$).

Equation \eqref{PDE} can be used to illustrate 3 different scales naturally appearing in our approach to soliton gases: the large number of micro-scale soliton-soliton interactions allows one to derive the NDR describing meso-scale DOS $u(z)$ and DOF $v(z)$, assuming that at these relatively large space-time scales the compact $\G^+$ and the spectral scaling function $\s(z)$, determining $u$, $v$, are virtually independent on $x$, $t$. Finally, equation \eqref{PDE} describes the macro-scale dynamics of $u(z;x,t)$ and $v(z;x,t)$ for non-equilibrium soliton gases, where the large scale $x$, $t$ dependence of $\G^+$ and $\s$ has to be taken into account.

Even though the existence and uniqueness of the solutions to the NDR \eqref{NDR}--\eqref{NDR2} were established in Theorem \ref{secondthm}, the explicit form of such solutions is known in a very few special cases, such as, for example, periodic gases, see \cite{TW}, or bound state soliton condensates, where $\s\equiv 0$ and $\G^+\subset a+{\rm i}\R^+$ for some $a\in\R$. fNLS soliton gases with such $\G^+$ are called bound state gases. This name reflects the fact that, as can be easily seen from \eqref{NDR}--\eqref{NDR2}, the effective speed
\eqref{v-eff}
for bound state gases is $s=-4 a$ for all $z\in\G^+$, i.e., all the elements (tracers) of the gas have the same effective velocity. It was observed in \cite{KT2021} that, for a bound state condensate, the DOS $u(z)$ is proportional to the $\frac{{\rm d}p}{{\rm d}z}$, where ${\rm d}p$ is the quasimomentum differential on the (limiting) hyperelliptic Riemann surface $\Rscr$ defined by $\G$. Similar results are valid for KdV soliton condensates with $u(z)$, $v(z)$ proportional to $\frac{{\rm d}p}{{\rm d}z}$, $\frac{{\rm d}q}{{\rm d}z}$, where ${\rm d}q$ is the quasienergy differential on $\Rscr$. For non-equilibrium KdV soliton condensates, it was proved in \cite{CERT} that the integro-differential kinetic equation for
the DOS $u=u(z;x,t)$ and DOF $v=v(z;x,t)$
reduces to the multi-phase KdV-Whitham modulation equations for the endpoints of $\G$ derived by
Flaschka, Forest and McLaughlin \cite{FFM} and Lax and Levermore \cite{LL83}.
Riemann problems for soliton condensates and explicit solutions for
the kinetic equation describing generalized rarefaction and dispersive shock waves were recently considered in \cite{CERT}.

Extension of these results of \cite{CERT} to the fNLS bound states condensate is trivial since, as it was mentioned earlier, $s(z;x,t)$ is constant. In the
present paper we provide the first known explicit solutions of the NDR \eqref{NDR}--\eqref{NDR2} for a
class of
non bound state fNLS condensates (i.e., $\s=0$), namely, for circular condensates, where the contour $\G^+$ is a collection of arcs (bands) located on an upper semicircle $|z|=\rho$, where $\rho$ is a real positive constant.

We then obtain the Whitham equations describing the $x$, $t$ dynamics of the endpoints of the arcs of $\G^+$,
as well as some explicit solutions for
the kinetic equation describing generalized rarefaction and dispersive shock waves.
We also analyze the kurtosis $\kappa$
for equilibrium and non-equilibrium circular fNLS gases.
Here the kurtosis is the fourth
normalized moment
\begin{align*}
\kappa=\frac{\braket{|\psi|^4}}{\braket{|\psi|^2}^2}
\end{align*}
of the probability density function (PDF) of the random wave amplitude $\psi$.
The fNLS evolution of
the so-called partially coherent waves, whose amplitude is
given by a slowly varying random function with a given
(e.g., Gaussian) statistics, was studied in \cite{CERTRS}. In particular, it was shown there that long time fNLS evolution of partially coherent waves leads to the doubling of the initial $	\kappa$, so that the initial $	\kappa_0=2$, corresponding to the Gaussian distribution, eventually becomes $	\kappa_\infty=4$, indicating ``fat tail'' distribution and, thus, potential presence of rogue waves. In Section \ref{sec-kurt}, we calculate the kurtosis $\kappa$ for genus one and zero circular condensates and show that $\kappa=2$ for a rarefaction wave and $\kappa>2$ for a dispersive shock wave.
Considering special families of genus one condensates in the limit of diminishing bands, we show that
any value $\kappa>2$ can be approached along certain trajectories in the corresponding parameter space.

The rest of the paper is organized as follows: in Section~\ref{sec-Main}, we present the main results of the paper (Theorems~\ref{th-main-1} and~\ref{th-main-2}). The proof of Theorem~\ref{th-main-1} is given in Section~\ref{s2}. In Section~\ref{s3}, we apply Theorem~\ref{th-main-1} to study the circular condensates in genus~0 and~1 cases, where the explicit formulae for the DOS and DOF are presented. In Section~\ref{sec-Riem-prob}, we study the modulational dynamics of circular condensates.
In Section~\ref{sec-kurt}, we analyze the kurtosis of the genus 0 and 1 circular condensates.

\subsection{Main results}
\label{sec-Main}

In this paper, we study the fNLS soliton condensate supported on a compact $\G^+\subset S^+$, where $S^+\subset \overline{\C^+}$ is the centered at $z=0$ semicircle of the radius $\rho>0$.

We consider $\G^+=\bigcup_{k=1}^N\{\rho {\rm e}^{{\rm i}\theta}\colon \theta\in [\alpha_{2k-2},\alpha_{2k-1}]\}$, where $0\leq \a_0 < \a_1<\cdots<\a_{2N-1}{\leq \pi}$,
a~collection
of $N\in \N$ closed non-degenerate arcs (bands) of $S^+$.
We call $z_j=\rho {\rm e}^{{\rm i}\a_j}$, $j=0,\dots,2N-1$, endpoints of the bands. %(branchpoints of the hyperelliptic Riemann surface )
The bands are interlaced with $N+1$ gaps lying on $S^+$, where the arcs from $z=\pm\rho$ to the nearest band are also considered as gaps, see Figure \ref{fig:my_label}. Any of the latter gaps are considered to be collapsed if the corresponding $\pm\rho\in\G^+$. We also take the reference measure $\l(w)$ in \eqref{NDR}--\eqref{NDR2} to be simply the arclength. The conformal map
\begin{gather}\label{conf-map-p}
p(z)=\hf\left(\frac z\rho+\frac \rho z\right)
\end{gather}
maps $\C^+$ onto $\C$ with two branch cuts from $\pm 1$ to $\pm\infty $ respectively,
where $p(S^+)=[-1,1]$. Let~$\mathfrak R_N$ denotes the hyperelliptic Riemann surface with the branch cuts on $(-\infty,-1]$, $[1,\infty)$ and on
\begin{gather*}
\hat \G^+=p\bigl(\G^+\bigr)\subset [-1,1],
\end{gather*}
where
$p(\G^+)=\bigcup_{k=1}^N[a_{2k-2},a_{2k-1}]$ and $a_j=p(z_j)=p(\rho\exp({\rm i}\a_j))$.
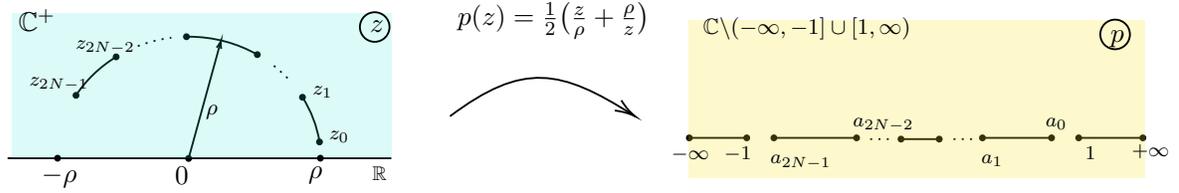
\begin{figure}[th]
    \centering
    \resizebox{0.99\textwidth}{!}{

\tikzset{every picture/.style={line width=0.75pt}} %set default line width to 0.75pt

\begin{tikzpicture}[x=0.75pt,y=0.75pt,yscale=-1,xscale=1]

%Straight Lines [id:da11043395115642407]
\draw    (61,166) -- (268.33,166) ;
%Shape: Arc [id:dp5857971438825527]
\draw  [draw opacity=0] (97.96,132.01) .. controls (103.47,123.68) and (110.89,116.57) .. (119.64,111.21) -- (158.8,167.13) -- cycle ; \draw    (97.96,132.01) .. controls (103.47,123.68) and (110.89,116.57) .. (119.64,111.21) ; \draw [shift={(119.64,111.21)}, rotate = 323.47] [color={rgb, 255:red, 0; green, 0; blue, 0 }  ][fill={rgb, 255:red, 0; green, 0; blue, 0 }  ][line width=0.75]      (0, 0) circle [x radius= 1.34, y radius= 1.34]   ; \draw [shift={(97.96,132.01)}, rotate = 308.77] [color={rgb, 255:red, 0; green, 0; blue, 0 }  ][fill={rgb, 255:red, 0; green, 0; blue, 0 }  ][line width=0.75]      (0, 0) circle [x radius= 1.34, y radius= 1.34]   ;
%Shape: Arc [id:dp48308113060202285]
\draw  [draw opacity=0] (157.63,100.31) .. controls (158.02,100.3) and (158.41,100.3) .. (158.8,100.3) .. controls (172.39,100.3) and (185.09,103.84) .. (195.91,109.99) -- (158.8,167.13) -- cycle ; \draw    (157.63,100.31) .. controls (158.02,100.3) and (158.41,100.3) .. (158.8,100.3) .. controls (172.39,100.3) and (185.09,103.84) .. (195.91,109.99) ; \draw [shift={(195.91,109.99)}, rotate = 25.21] [color={rgb, 255:red, 0; green, 0; blue, 0 }  ][fill={rgb, 255:red, 0; green, 0; blue, 0 }  ][line width=0.75]      (0, 0) circle [x radius= 1.34, y radius= 1.34]   ; \draw [shift={(157.63,100.31)}, rotate = 3.52] [color={rgb, 255:red, 0; green, 0; blue, 0 }  ][fill={rgb, 255:red, 0; green, 0; blue, 0 }  ][line width=0.75]      (0, 0) circle [x radius= 1.34, y radius= 1.34]   ;
%Shape: Arc [id:dp5615899707798235]
\draw  [draw opacity=0] (220.31,133.04) .. controls (224.93,140.31) and (228.12,148.48) .. (229.51,157.2) -- (158.8,167.13) -- cycle ; \draw    (220.31,133.04) .. controls (224.93,140.31) and (228.12,148.48) .. (229.51,157.2) ; \draw [shift={(229.51,157.2)}, rotate = 75.29] [color={rgb, 255:red, 0; green, 0; blue, 0 }  ][fill={rgb, 255:red, 0; green, 0; blue, 0 }  ][line width=0.75]      (0, 0) circle [x radius= 1.34, y radius= 1.34]   ; \draw [shift={(220.31,133.04)}, rotate = 62.88] [color={rgb, 255:red, 0; green, 0; blue, 0 }  ][fill={rgb, 255:red, 0; green, 0; blue, 0 }  ][line width=0.75]      (0, 0) circle [x radius= 1.34, y radius= 1.34]   ;
%Shape: Arc [id:dp36745803345523487]
\draw  [draw opacity=0][dash pattern={on 0.84pt off 2.51pt}] (131.32,105.41) .. controls (137.04,103.19) and (143.13,101.63) .. (149.49,100.86) -- (158.8,167.13) -- cycle ; \draw  [dash pattern={on 0.84pt off 2.51pt}] (131.32,105.41) .. controls (137.04,103.19) and (143.13,101.63) .. (149.49,100.86) ;
%Shape: Arc [id:dp11470689402830425]
\draw  [draw opacity=0][dash pattern={on 0.84pt off 2.51pt}] (204.84,116) .. controls (208.03,118.51) and (210.98,121.28) .. (213.66,124.27) -- (158.8,167.13) -- cycle ; \draw  [dash pattern={on 0.84pt off 2.51pt}] (204.84,116) .. controls (208.03,118.51) and (210.98,121.28) .. (213.66,124.27) ;
%Straight Lines [id:da16372735936793958]
\draw    (158.8,166.13) -- (175.8,104.26) ;
\draw [shift={(176.33,102.33)}, rotate = 105.37] [color={rgb, 255:red, 0; green, 0; blue, 0 }  ][line width=0.75]    (4.37,-1.32) .. controls (2.78,-0.56) and (1.32,-0.12) .. (0,0) .. controls (1.32,0.12) and (2.78,0.56) .. (4.37,1.32)   ;
\draw [shift={(158.8,166.13)}, rotate = 285.37] [color={rgb, 255:red, 0; green, 0; blue, 0 }  ][fill={rgb, 255:red, 0; green, 0; blue, 0 }  ][line width=0.75]      (0, 0) circle [x radius= 1.34, y radius= 1.34]   ;
%Curve Lines [id:da5763794809339404]
\draw    (300,143.33) .. controls (339.6,113.63) and (359.92,117.91) .. (398.82,142.58) ;
\draw [shift={(400,143.33)}, rotate = 212.56] [color={rgb, 255:red, 0; green, 0; blue, 0 }  ][line width=0.75]    (10.93,-3.29) .. controls (6.95,-1.4) and (3.31,-0.3) .. (0,0) .. controls (3.31,0.3) and (6.95,1.4) .. (10.93,3.29)   ;
%Shape: Rectangle [id:dp5634612549455493]
\draw  [draw opacity=0][fill={rgb, 255:red, 33; green, 235; blue, 212 }  ,fill opacity=0.16 ][line width=0.75]  (63,87.67) -- (268.33,87.67) -- (268.33,166) -- (63,166) -- cycle ;
%Straight Lines [id:da040773734353770186]
\draw    (429.33,155) -- (460.33,155) ;
\draw [shift={(460.33,155)}, rotate = 0] [color={rgb, 255:red, 0; green, 0; blue, 0 }  ][fill={rgb, 255:red, 0; green, 0; blue, 0 }  ][line width=0.75]      (0, 0) circle [x radius= 1.34, y radius= 1.34]   ;
\draw [shift={(429.33,155)}, rotate = 0] [color={rgb, 255:red, 0; green, 0; blue, 0 }  ][fill={rgb, 255:red, 0; green, 0; blue, 0 }  ][line width=0.75]      (0, 0) circle [x radius= 1.34, y radius= 1.34]   ;
%Straight Lines [id:da6802437508579728]
\draw    (639.67,155) -- (674.33,155) ;
\draw [shift={(674.33,155)}, rotate = 0] [color={rgb, 255:red, 0; green, 0; blue, 0 }  ][fill={rgb, 255:red, 0; green, 0; blue, 0 }  ][line width=0.75]      (0, 0) circle [x radius= 1.34, y radius= 1.34]   ;
\draw [shift={(639.67,155)}, rotate = 0] [color={rgb, 255:red, 0; green, 0; blue, 0 }  ][fill={rgb, 255:red, 0; green, 0; blue, 0 }  ][line width=0.75]      (0, 0) circle [x radius= 1.34, y radius= 1.34]   ;
%Straight Lines [id:da26312804044861404]
\draw    (475,155) -- (519.67,155) ;
\draw [shift={(519.67,155)}, rotate = 0] [color={rgb, 255:red, 0; green, 0; blue, 0 }  ][fill={rgb, 255:red, 0; green, 0; blue, 0 }  ][line width=0.75]      (0, 0) circle [x radius= 1.34, y radius= 1.34]   ;
\draw [shift={(475,155)}, rotate = 0] [color={rgb, 255:red, 0; green, 0; blue, 0 }  ][fill={rgb, 255:red, 0; green, 0; blue, 0 }  ][line width=0.75]      (0, 0) circle [x radius= 1.34, y radius= 1.34]   ;
%Straight Lines [id:da6905734130317058]
\draw    (543.67,155.67) -- (564.33,155.67) ;
\draw [shift={(564.33,155.67)}, rotate = 0] [color={rgb, 255:red, 0; green, 0; blue, 0 }  ][fill={rgb, 255:red, 0; green, 0; blue, 0 }  ][line width=0.75]      (0, 0) circle [x radius= 1.34, y radius= 1.34]   ;
\draw [shift={(543.67,155.67)}, rotate = 0] [color={rgb, 255:red, 0; green, 0; blue, 0 }  ][fill={rgb, 255:red, 0; green, 0; blue, 0 }  ][line width=0.75]      (0, 0) circle [x radius= 1.34, y radius= 1.34]   ;
%Straight Lines [id:da7457553690145493]
\draw    (587.67,155) -- (625,155) ;
\draw [shift={(625,155)}, rotate = 0] [color={rgb, 255:red, 0; green, 0; blue, 0 }  ][fill={rgb, 255:red, 0; green, 0; blue, 0 }  ][line width=0.75]      (0, 0) circle [x radius= 1.34, y radius= 1.34]   ;
\draw [shift={(587.67,155)}, rotate = 0] [color={rgb, 255:red, 0; green, 0; blue, 0 }  ][fill={rgb, 255:red, 0; green, 0; blue, 0 }  ][line width=0.75]      (0, 0) circle [x radius= 1.34, y radius= 1.34]   ;
%Straight Lines [id:da7312633751464399]
\draw  [dash pattern={on 0.84pt off 2.51pt}]  (527,155.67) -- (538,155.67) ;
%Straight Lines [id:da15520767895756094]
\draw  [dash pattern={on 0.84pt off 2.51pt}]  (572,155.67) -- (583,155.67) ;
%Shape: Rectangle [id:dp17095356908229942]
\draw  [draw opacity=0][fill={rgb, 255:red, 248; green, 231; blue, 28 }  ,fill opacity=0.22 ][line width=0.75]  (429,91.33) -- (675.67,91.33) -- (675.67,176.67) -- (429,176.67) -- cycle ;

\draw [shift={(88,166)}, rotate = 0] [color={rgb, 255:red, 0; green, 0; blue, 0 }  ][fill={rgb, 255:red, 0; green, 0; blue, 0 }  ][line width=0.75]      (0, 0) circle [x radius= 1.34, y radius= 1.34]   ;
\draw [shift={(230,166)}, rotate = 0] [color={rgb, 255:red, 0; green, 0; blue, 0 }  ][fill={rgb, 255:red, 0; green, 0; blue, 0 }  ][line width=0.75]      (0, 0) circle [x radius= 1.34, y radius= 1.34]   ;

% Text Node
\draw (256.67,169.07) node [anchor=north west][inner sep=0.75pt]  [font=\scriptsize]  {$\mathbb{{\textstyle R}}$};
% Text Node
\draw (233.33,149.73) node [anchor=north west][inner sep=0.75pt]  [font=\scriptsize]  {$z _{0}$};
% Text Node
\draw (224.33,125.73) node [anchor=north west][inner sep=0.75pt]  [font=\scriptsize]  {$z _{1}$};
% Text Node
%\draw (196.33,99.73) node [anchor=north west][inner sep=0.75pt]  [font=\scriptsize]  {$z _{2k}$};
%% Text Node
%\draw (155.33,85.73) node [anchor=north west][inner sep=0.75pt]  [font=\scriptsize]  {$z _{2k+1}$};
% Text Node
\draw (71.33,120.73) node [anchor=north west][inner sep=0.75pt]  [font=\scriptsize]  {$z _{2N-1}$};
% Text Node
\draw (96.33,100.73) node [anchor=north west][inner sep=0.75pt]  [font=\scriptsize]  {$z _{2N-2}$};
% Text Node
\draw (150,169.07) node [anchor=north west][inner sep=0.75pt]    {$0$};
\draw (78,169.07) node [anchor=north west][inner sep=0.75pt]    {$-\rho$};
\draw (222,169.07) node [anchor=north west][inner sep=0.75pt]    {$\rho$};
% Text Node
\draw (166.67,135.73) node [anchor=north west][inner sep=0.75pt]  [font=\footnotesize]  {$\rho $};
% Text Node
\draw (302.67,80.07) node [anchor=north west][inner sep=0.75pt]  [font=\normalsize]  {$p(z) =\frac{1}{2}\bigl(\frac{z}{\rho } +\frac{\rho }{z}\bigr)$};
% Text Node
\draw (620.33,142.73) node [anchor=north west][inner sep=0.75pt]  [font=\scriptsize]  {$a_{0}$};
% Text Node
\draw (585.67,162.73) node [anchor=north west][inner sep=0.75pt]  [font=\scriptsize]  {$a_{1}$};
% Text Node
%\draw (561,142.73) node [anchor=north west][inner sep=0.75pt]  [font=\scriptsize]  {$a_{2k}$};
%% Text Node
%\draw (539,162.73) node [anchor=north west][inner sep=0.75pt]  [font=\scriptsize]  {$a_{2k+1}$};
% Text Node
\draw (516.33,142.73) node [anchor=north west][inner sep=0.75pt]  [font=\scriptsize]  {$a_{2N-2}$};
% Text Node
\draw (471.67,162.73) node [anchor=north west][inner sep=0.75pt]  [font=\scriptsize]  {$a_{2N-1}$};
% Text Node
\draw (641.67,158.4) node [anchor=north west][inner sep=0.75pt]  [font=\scriptsize]  {$1$};
% Text Node
\draw (446.83,158.4) node [anchor=north west][inner sep=0.75pt]  [font=\scriptsize]  {$-1$};
% Text Node
\draw (418.33,158.4) node [anchor=north west][inner sep=0.75pt]  [font=\scriptsize]  {$-\infty $};
% Text Node
\draw (667,157.07) node [anchor=north west][inner sep=0.75pt]  [font=\scriptsize]  {$+\infty $};
% Text Node
\draw (65.33,84.07) node [anchor=north west][inner sep=0.75pt]    {$\mathbb{C}^{+}$};
\draw (255.33,90.07) node [anchor=north west][inner sep=0.75pt]    {{$z$}};
\node[circle,draw=black, inner sep=0pt,minimum size=12pt] at (259.33,95.07) {};
% Text Node
\draw (435.33,88.07) node [anchor=north west][inner sep=0.75pt]  [font=\footnotesize]  {$\mathbb{C} \backslash {\textstyle ( -\infty ,-1] \cup [ 1,\infty )}$};
\draw (655.33,95.07) node [anchor=north west][inner sep=0.75pt]    {$p$};
\node[circle,draw=black, inner sep=0pt,minimum size=12pt] at (659.33,99.07) {};

\end{tikzpicture}

    }
    \caption{Illustration of the conformal map $p(z)$ from the upper half plane with $n$ circular slits (namely, $\bigcup_{k=1}^N\{\rho e^{i\theta}\colon \theta\in [\alpha_{2k-2},\alpha_{2k-1}]\}$) to $\C\backslash (-\infty,-1]\cup(1,\infty]$ with $N$ slits (namely,    $\bigcup_{k=1}^N[a_{2k-2},a_{2k-1}]$, where $a_j=p(\rho\exp(i\a_j))$) within the interval $[-1,1]$. In the left figure, $z_j=\rho e^{i\alpha_j}$ indicate the branch points. The genus of $\mathfrak R_N$, see the right plane,  is $N$.}
    \label{fig:my_label}
\end{figure} 

Denote by $\hat u(p)$, $\hat v(p)$ solutions of singular integral equations
\begin{gather}\label{Hilb-p}
\pi H[\hat u](p):=\int_{\hat \G^+}\frac{\hat u(q){\rm d}q}{q-p}=-p, \qquad \pi H[\hat v](p)=4\rho\bigl(2p^2-1\bigr) \qquad \text{on}\ \hat\G^+ ;
\end{gather}
here $H$ denotes the finite Hilbert transform (FHT) on $\hat\G^+$. It is straightforward to show (see, for example, \cite{KT2021}) that
\be\label{hat-uv}
\hat u(p)= \frac{P(p)}{R(p)}, \qquad
\hat v(p)=\frac{Q(p)}{ R(p)},
\ee
where $P$, $Q$ are polynomials of degree $N+1$ and $N+2$ respectively and \smash{$R(p)=\prod_j(p-p(z_j))^\hf$} taken over all the endpoints of $\hat \G^+$ and normalized by $R(p)\sim p^N$ as $p\ra\infty$ in $\C$.
However, $P$, $Q$ are not uniquely defined by \eqref{Hilb-p}, since the FHT $H$ has an $N$-dimensional kernel. The following theorem establishes solutions of the NDR \eqref{NDR}--\eqref{NDR2} in terms of $\hat u$, $\hat v$.
\begin{Theorem}\label{th-main-1}
The NDR \eqref{NDR}--\eqref{NDR2} for fNLS circular soliton condensate $($i.e., with $\sigma\equiv 0)$ have solutions
\begin{gather}%\label{NDRinMain}
u(z)= \hat u(p(z))= \frac{P(p(z))}{ R(p(z))}, \qquad \text{where} \quad \pi H[\hat u]=-p \ \text{on}\ \hat\G^+=p(\G^+),\nonumber\\
v(z)= \hat v(p(z))=\frac{ Q(p(z))}{ R(p(z))},\qquad \text{where} \quad \pi H[\hat v]=4\rho(2p^2-1)
\ \text{on}\ \hat\G^+=p(\G^+). \label{NDRinMain2}
\end{gather}
Here $P(p)$, $Q(p)$ are polynomials of degrees $N+1$, $N+2$ respectively, $p\in\R$, \smash{$\frac{\hat u(p){\rm d}p}{\sqrt{1-p^2}}$}, \smash{$\frac{\hat v(p){\rm d}p}{\sqrt{1-p^2}}$} are
second kind meromorphic differentials on $\mathcal{R}_N$ and
$u$, $v$ satisfy the conditions
\begin{gather} \label{gap-cond}
\int_{c_j} u(z)\,{\rm d}\arg z=0, \qquad
\int_{c_j} v(z)\,{\rm d}\arg z=0
\end{gather}
for all
$j=0,\dots,N$, where $c_j=\{\rho {\rm e}^{{\rm i}\theta}\colon \theta\in (\alpha_{2j-1},\alpha_{2j})\}$, $\alpha_{-1}=0$, and $\alpha_{2N}=\pi$, denote the gaps on $S^+\setminus \G^+$ where
%the semicircle
$S^+=\{|z|=\rho,\, 0\leq \arg z\leq \pi\}$ $($see Figure~$\ref{fig:my_label})$. In the case when an endpoint $\pm\rho\in\G^+$ and so the corresponding gap $c_j$ collapses, the corresponding integral conditions in~\eqref{gap-cond} should be replaced by $u(\pm\rho)=0$,
$v(\pm\rho)=0$. Also, $u(z),v(z)\in\R$ on $\G^+$ and $u>0$ on $\G^+\setminus\R$.
\end{Theorem}

According to Theorem \ref{secondthm}, a solution to NDR is unique.
We now state the second main theorem about non-uniform circular soliton condensate for the fNLS,
where the endpoints ${z_j=\rho {\rm e}^{{\rm i}\a_j}}$ of the bands on $S^+$ can depend on $x$, $t$.
It is governed by the equation \eqref{PDE}
and reflect large scale changes of $u$, $v$.

Let $\vec \a(x,t)=(\a_0(x,t),\dots,\a_{2N-1}(x,t))$ denote the vector of $x$, $t$ dependent endpoints of $\G^+$. Assuming smoothness of $u(z;x,t)$, $v(z;x,t)$ and
$\vec \a(x,t)$, and using Theorem \ref{th-main-1} , we follow the approach of \cite{FFM} to show that
the continuity equation \eqref{PDE} can be written as a Whitham type equations on the evolution of $\vec \a(x,t)$. A similar result for the KdV soliton condensates was obtained in \cite{CERT}.

\begin{Theorem}\label{th-main-2}
If the fNLS circular soliton condensate that described in Theorem $\ref{th-main-1}$, with ${\G^+\cap\R=\varnothing}$, is non-equilibrium, then \eqref{PDE} is equivalent to the system of modulation $($Whitham$)$ equations
% for endpoints $\a_j$ are
given by
\begin{gather}\label{Whit-alp}
\partial_t\a_j+V_j(\vec\a)\partial_x\a_j=0, \qquad j=0,\dots, 2N-1,
\end{gather}
where $z_j=\rho {\rm e}^{{\rm i}\a_j}$ and
%$\vec\a=(\a_0,\dots,\a_{2n-1})$ and
\begin{gather}\label{mod-vel}
V_j(\vec\a)=\frac{Q(\cos\a_j)}{P(\cos\a_j)}
\end{gather}
are bounded velocities, with $P$, $Q$ as in Theorem~$\ref{th-main-1}$.
\end{Theorem}

\begin{proof}
Let us first obtain \eqref{Whit-alp}	 from \eqref{PDE}. According to Theorem \ref{th-main-1}, \eqref{PDE} can be written as
\begin{gather}\label{step1}
\left(\frac {P(p)}{R(p)}\right)_t+\left(\frac {Q(p)}{R(p)}\right)_x=0,
\end{gather}
which should be valid for all $p\in\hat\G^+$. Denoting $T=R^2$,
equation \eqref{step1} can be written as
\begin{gather}\label{KdVkin1}
2T(P_t+Q_x)=PT_t+QT_x \qquad {\rm or} \qquad 2(P_t+Q_x)=P(\log T)_t+Q(\log T)_x.
\end{gather}
%where $S(z)=R^2(z)=\prod_{j=0}^N(z^2-a_j^2)$.
Since
\begin{gather*}
(\log T)_r=-\sum_{j=0}^{2N-1}\frac{(a_j)_r}{p-a_j},
\end{gather*}
where $r$ is either $x$ or $t$,
the second equation \eqref{KdVkin1} yields
\begin{gather*}%\label{KdV-Whit}
P(a_j)(a_j)_t+Q(a_j)(a_j)_x=0, \qquad j=0,1,\dots,2N-1,
\end{gather*}
if we take limit $p\ra a_j$. Given $a_j=\cos \a_j$, the latter equation implies \eqref{Whit-alp}--\eqref{mod-vel}. Note that all the zeros of $P(p)$ are on the gaps and, thus, the velocities $V_j$ defined by \eqref{mod-vel} are bounded.

Assume now that the evolution of the endpoints satisfies
\eqref{Whit-alp} or, equivalently,
\begin{gather*}%\label{Whit-p}
\partial_t a_j+V_j(\vec a)\partial_x a_j=0, \qquad j=0,\dots, 2N-1,
\end{gather*}
where $\vec a=(a_0,\dots, a_{2N-1})$.
Then, following \cite{FFM}, we observe that the only poles of the differential $\Omega=\partial_t\hat u{\rm d}p+\partial_x\hat v{\rm d}p$ on $\mathfrak R_N$ are the second-order poles at each $a_j$. But the modulation equations~\eqref{Whit-alp} show that the principal parts of $\Omega$ at each $p=a_j$ is zero, that is, $\Omega$ is a holomorphic differential.
Since all the gap integrals (and, thus, the B-periods) of a holomorphic differential~$\Omega$ are zeros, we obtain the kinetic equation $\partial_t\hat u+\partial_x\hat v=0$, which is equivalent to \eqref{PDE}.
\end{proof}

The modulation equations
\eqref{Whit-alp}--\eqref{mod-vel} form a strictly hyperbolic system of first-order quasilinear PDEs provided that all the branchpoints are distinct (otherwise it will be just hyperbolic). This system is in the diagonal (Riemann) form with all the coefficients (velocities) being real. The Cauchy data for this system consists of $\vec\a(x,0)$. The system has a unique local (classic) real solution provided $\vec\a(x,0)$ is of $C^1$ class, see \cite[Theorem 7.8.1]{DafBook}, and real and the velocities~$V_j(\vec\a)$ are smooth and real. Thus, the fNLS circular gas is (at least locally) preserved under the evolution described by the kinetic equation.

\begin{Remark}\label{rem-break}
As it is well known, systems of hyperbolic equations may develop singularities in the $x$, $t$ plane, which, in the case of modulation equations \eqref{Whit-alp}--\eqref{mod-vel}, lead to collapse of a~band or a gap, or to appearance of a new ``double point'' that will open into a band or gap. In any case, at a point of singularity $($also known as a breaking point$)$, two or more endpoints from~$\vec a(x,t)$ collide or a new pair$($s$)$ of collapsed double points appear, so that the Riemann surface $\mathfrak R_N$ develop a singularity.
In this paper we do not intend to discuss details of transition of the circular condensate between regions of different genera while passing through a breaking point.
However, we would like to mention that
since the differentials $\hat u(p){\rm d}p$, $\hat v(p){\rm d}p$
%, as well as their $x, t$ derivatives,
are imaginary normalized differentials, they undergo a continuous transition through breaking points, see~{\rm \cite{BT14}}. Thus, fNLS circular condensate is preserved under the kinetic equation evolution through breaking points $($change of genus$)$. Some examples of such evolution can be found in Section~$\ref{sec-Riem-prob}$.\looseness=-1
\end{Remark}

\section{Proof of Theorem \ref{th-main-1}}\label{s2}
In the particular case of a genus zero circular condensate where the point $z=\rho$ is on the band, the NDR where solved in \cite{ET2020}.
The proof of Theorem \ref{th-main-1} presented below in a sense resembles the proof of the solution to the NDR for a bound state fNLS soliton condensate from
\cite{KT2021}.

The conformal map \eqref{conf-map-p} has the inverse $z=\rho\bigl(p+\sqrt{p^2-1}\bigr)$ or $z=\rho {\rm e}^{{\rm i}\xi}$, where $p=\cos \xi$. In the variables $\xi$, $\th$, where $w=\rho {\rm e}^{{\rm i}\th}$, each NDR equation in \eqref{NDR} with $\s\equiv 0$ can be written as
\begin{gather}\label{NDR-xi}
-\r\int_{\tilde\G^+}\log\Biggl|\frac{\sin \frac{\x-\th}{2}}{\sin \frac{\x+\th}{2}}\Biggr|\psi_j(\th)\,{\rm d}\th=
\phi_j(\xi),\qquad j=1,2,
\end{gather}
where $\psi_1(\xi)=u\bigl(\r {\rm e}^{{\rm i}\xi}\bigr)$, $\psi_2(\xi)=v\bigl(\r {\rm e}^{{\rm i}\xi}\bigr)$, $\tilde\G^+$ is the preimage of $\G^+$ under the map $z=\rho {\rm e}^{{\rm i}\xi}$,
\begin{gather*}%\label{RHS12}
\phi_1(\xi)= \r\sin\x, \qquad
\phi_2(\xi)= -4\rho^2 \sin\xi\cos\xi,
\end{gather*}
and the integration in $\hat\G^+\subset \R$ goes in the negative direction. Here and henceforth, we always assume that the integration over $\hat\G^+$ goes in the positive direction, and, therefore, change the sign in the left-hand side of \eqref{NDR-xi}.
Since
\begin{gather*}
\frac{{\rm d}}{{\rm d}\x}\log\Biggl|\frac{\sin \frac{\x-\th}{2}}{\sin \frac{\x+\th}{2}}\Biggr|=\hf\biggl[ \cot \frac{\x-\th}{2}- \cot \frac{\x+\th}{2} \biggr]=
\frac{\sin\th}{\cos\th-\cos\x},
\end{gather*}
differentiation in $\x$ of \eqref{NDR-xi} yields
\begin{gather*}%\label{Hilb-1}
\r\int_{\tilde\G^+}\frac{\psi_j(\th)\sin\th}{\cos\th-\cos\x}\,{\rm d}\th=
-\frac{{\rm d}}{{\rm d}\xi}\phi_j(\xi),\qquad j=1,2,
\end{gather*}
or
\begin{align}\label{Hilb-2}
\pi H[\hat u](p)=-p,\qquad
\pi H[\hat v](p)=4\rho^{} \bigl(2p^2-1\bigr),
\end{align}
where $\hat u(q)=\psi_1(\th)$, $\hat v(q)=\psi_2(\th)$ and $q=\cos\th$.

Inversion of the first FHT in \eqref{Hilb-2} has the form \eqref{hat-uv}, where the degree $N+1$ polynomial~$P(p)$ is defined up to a kernel of $H$ acting on $\hat\G^+$.
As it is well known (and can be easily verified), this kernel is an $N$-dimensional space that consists of functions $\frac{K(p)}{R(p)}$, where $K(p)$ is an arbitrary
polynomial of degree $N-1$.

We now prove that the $N$ unknown coefficients of $K(p)$ are uniquely defined by conditions~\eqref{gap-cond} for $\hat u$. We start with proving that \smash{$\frac{\hat u(p){\rm d}p}{\sqrt{1-p^2}}$} has zero residue at $p=\infty$, i.e., it~is~a~second kind meromorphic differential on $\mathfrak{R}_N$. Indeed, substituting $\hat u(p)=\frac{P(p)}{R(p)}$ into \eqref{Hilb-2} and calculating the $H[\hat u]$ through the residue at $p=\infty$, we obtain
\begin{gather}\label{coeff-u}
a=-{\rm i}, \qquad b=\frac {\rm i}2\sum_{j=0}^{2N-1}a_j,
\end{gather}
where $P(p)={\frac{1}{\pi}\bigl(ap^{N+1}+bp^N+\cdots\bigr)}$ and $a_j=\cos\a_j$ are the endpoints of $\hat\G^+$. Thus
\begin{gather*}
\frac{\hat u(p){\rm d}p}{\sqrt{1-p^2}} =\frac{1-\frac{\sum_ja_j}{2p}+\cdots}{\sqrt{1-p^{-2}}\prod_j\bigl(1-\frac{a_j}{p}\bigr)^\hf}\,{\rm d}p=\bigl(1+ O\bigl(p^{-2}\bigr)\bigr){\rm d}p,
\end{gather*}
which complete the argument.

Now equations \eqref{gap-cond} for $u$ can be written as
\begin{gather}\label{gap-cond-1}
\int_{a_{2j-1}}^{a_{2j}}\frac{\hat u(p){\rm d}p}{\sqrt{1-p^2}}=0, \qquad j=0,\dots,N,
\end{gather}
with $a_{-1}=1$, $a_{2N}=-1$ and the corresponding integral in \eqref{gap-cond-1} should be replaced by ${\hat u(\pm 1)=0}$ if
$\pm 1\in\hat\G^+$ respectively.
The fact that \smash{$\Omega_1=\frac{\hat u(p){\rm d}p}{\sqrt{1-p^2}}$} is a second kind differential implies that one of the equations~\eqref{gap-cond-1} is a tautology and so the remaining $n$ conditions simply define a~normalization of $\Omega_1$. In particular, if all except one gap are $A$-cycles, then \eqref{gap-cond-1} implies that $\Omega_1$ is an $A$-normalized meromorphic differential. Thus, equations \eqref{gap-cond-1} always have a unique solution. The cases $\pm 1\in\hat\G^+$ can be treated as limits of small closing gaps from $\pm 1$ to the nearest endpoint.

To complete our arguments for $u(z)$, we need to show that conditions \eqref{gap-cond} for $u$ must be satisfied. In this proof we follow the arguments of~\cite[Theorem~6.1]{KT2021}, where similar conditions were derived for the bound state fNLS soliton condensate, i.e., when all the bands were situated on the imaginary axis. The idea of the proof is related to the fact that the solution $u(z)$ to~\eqref{NDR} is the density of the equilibrium measure for the corresponding Green's energy~\cite{KT2021}.
If $u$ is such an equilibrium density then the Green's potential
$G[u]:=\int_{\G^+}\log\bigl|\frac{z-\bar w }{z-w}\bigr|u(w)|{\rm d}w|$ of $u$ should be continuous at every regular point of $\G^+$, see \cite{ST}. In our case, all points of $\G^+$ are regular and so $G[u]$ must be continuous in $\C$.

Equations \eqref{gap-cond-1} imply
that there exists at least one zero in each gaps (intervals). Since there are $n+1$ gaps and the degree of $P(p)$ is $N+1$, we conclude that each gap has exactly one root of $P(p)$ and, so, all the roots of the polynomial $P(p)$ are real. Thus, according to \eqref{coeff-u}, $P(p)$ is purely imaginary on $\R$ and, so, $u>0$ on $\G^+$.

Similar arguments hold for the solution $v(z)$ of the second NDR \eqref{NDRinMain2}.
For example, representing $Q(p)=\frac{1}{\pi}\bigl(ap^{N+2}+bp^{N+1}+cp^N+\cdots\bigr)$, it follows from \eqref{Hilb-p} that $a=-8{\rm i}\r $, $b=-\frac a2 \sum_j a_j$ and $c=4{\rm i}\r +{\rm i}\r \bigl(\sum_{j}a_j^2-6\sum_{j<k}a_ja_k\bigr)$. So, $a,b,c\in {\rm i}\R$. Equations~\eqref{coeff-u} imply that $N+1$ roots are real. Thus it follows that all the roots of $Q$ are real.

\section{Genus 0 and genus 1 circular condensates}\label{s3}
In this section, the results of Theorem \ref{th-main-1} are applied to two simple cases: the genus 0 and~1 circular soliton condensate.
We remind that by genus we mean the genus of the Riemann surface~$\mathfrak R_N$, $N=1,2,3$, which, generically, is equal to the number of bands $N$ (or the number of gaps minus one), see Figure~\ref{fig:my_label}.
However, since in this section we always have $\alpha_0=0$ and $\alpha_{2N-1}=\pi$, two gaps are collapsed and the genus of the Riemann surface is reduced to $N=0,1$.

As before, we denote
\[
a_j=p\bigl(\r {\rm e}^{{\rm i}\a_j}\bigr)=\cos(\alpha_j), \qquad j=1,2,3,4.
\]
 The support $\G^+$ for the genus one circular condensate is illustrated by Figure~\ref{fig:genus one master}. In what follows, we will not mention $\a_0$ and~$\a_{2N+3}$ since they are always fixed.
In both cases, the exact solutions to the NDR equations can be explicitly represented with the help of complete elliptic integrals. For higher genus situation, the solution to the NDR equations can be represented using hyperelliptic integrals.

\begin{figure}[t]
\centering
\resizebox{0.44\textwidth}{!}{\tikzset{every picture/.style={line width=0.75pt}} %set default line width to 0.75pt        

\begin{tikzpicture}[x=0.75pt,y=0.75pt,yscale=-1,xscale=1]
	
	%Straight Lines [id:da9415440149591279] 
	\draw    (207.33,163) -- (477,163.33) ;
	%Shape: Arc [id:dp3332794895139495] 
	\draw  [draw opacity=0] (268.03,109.83) .. controls (280.72,93.61) and (299.22,81.84) .. (321.1,77.99) .. controls (341.66,74.38) and (361.8,78.36) .. (378.64,87.91) -- (336.08,163.19) -- cycle ; \draw    (268.03,109.83) .. controls (280.72,93.61) and (299.22,81.84) .. (321.1,77.99) .. controls (341.66,74.38) and (361.8,78.36) .. (378.64,87.91) ; \draw [shift={(378.64,87.91)}, rotate = 25.72] [color={rgb, 255:red, 0; green, 0; blue, 0 }  ][fill={rgb, 255:red, 0; green, 0; blue, 0 }  ][line width=0.75]      (0, 0) circle [x radius= 3.35, y radius= 3.35]   ; \draw [shift={(268.03,109.83)}, rotate = 312.11] [color={rgb, 255:red, 0; green, 0; blue, 0 }  ][fill={rgb, 255:red, 0; green, 0; blue, 0 }  ][line width=0.75]      (0, 0) circle [x radius= 3.35, y radius= 3.35]   ;
	%Shape: Arc [id:dp5363284559358221] 
	\draw  [draw opacity=0][line width=1.5]  (404.06,109.75) .. controls (412.6,120.61) and (418.66,133.67) .. (421.22,148.23) .. controls (422.04,152.93) and (422.47,157.61) .. (422.53,162.23) -- (336.08,163.19) -- cycle ; \draw [line width=1.5]    (404.06,109.75) .. controls (412.6,120.61) and (418.66,133.67) .. (421.22,148.23) .. controls (422.04,152.93) and (422.47,157.61) .. (422.53,162.23) ; \draw [shift={(422.53,162.23)}, rotate = 85.36] [color={rgb, 255:red, 0; green, 0; blue, 0 }  ][fill={rgb, 255:red, 0; green, 0; blue, 0 }  ][line width=1.5]      (0, 0) circle [x radius= 2.61, y radius= 2.61]   ; \draw [shift={(404.06,109.75)}, rotate = 55.82] [color={rgb, 255:red, 0; green, 0; blue, 0 }  ][fill={rgb, 255:red, 0; green, 0; blue, 0 }  ][line width=1.5]      (0, 0) circle [x radius= 2.61, y radius= 2.61]   ;
	%Shape: Circle [id:dp38847979876681604] 
	\draw  [fill={rgb, 255:red, 0; green, 0; blue, 0 }  ,fill opacity=1 ] (334.67,162.67) .. controls (334.67,161.56) and (335.56,160.67) .. (336.67,160.67) .. controls (337.77,160.67) and (338.67,161.56) .. (338.67,162.67) .. controls (338.67,163.77) and (337.77,164.67) .. (336.67,164.67) .. controls (335.56,164.67) and (334.67,163.77) .. (334.67,162.67) -- cycle ;
	%Shape: Arc [id:dp030040374752331678] 
	\draw  [draw opacity=0][dash pattern={on 0.84pt off 2.51pt}] (379.62,88.86) .. controls (388.96,94.22) and (397.27,101.3) .. (404.06,109.75) -- (336.65,163.92) -- cycle ; \draw [color={rgb, 255:red, 74; green, 144; blue, 226 }  ,draw opacity=1 ][dash pattern={on 0.84pt off 2.51pt}] [dash pattern={on 0.84pt off 2.51pt}]  (379.62,88.86) .. controls (388.96,94.22) and (397.27,101.3) .. (404.06,109.75) ;  
	%Shape: Arc [id:dp6483332818512029] 
	\draw  [draw opacity=0][line width=1.5]  (268.61,109.31) .. controls (281.31,93.08) and (299.81,81.31) .. (321.69,77.47) .. controls (342.25,73.85) and (362.38,77.83) .. (379.23,87.38) -- (336.67,162.67) -- cycle ; \draw [line width=1.5]    (268.61,109.31) .. controls (281.31,93.08) and (299.81,81.31) .. (321.69,77.47) .. controls (342.25,73.85) and (362.38,77.83) .. (379.23,87.38) ; \draw [shift={(379.23,87.38)}, rotate = 25.72] [color={rgb, 255:red, 0; green, 0; blue, 0 }  ][fill={rgb, 255:red, 0; green, 0; blue, 0 }  ][line width=1.5]      (0, 0) circle [x radius= 1.74, y radius= 1.74]   ; \draw [shift={(268.61,109.31)}, rotate = 312.11] [color={rgb, 255:red, 0; green, 0; blue, 0 }  ][fill={rgb, 255:red, 0; green, 0; blue, 0 }  ][line width=1.5]      (0, 0) circle [x radius= 1.74, y radius= 1.74]   ;
	%Straight Lines [id:da44229970288016984] 
	\draw [color={rgb, 255:red, 155; green, 155; blue, 155 }  ,draw opacity=1 ]   (336.67,162.67) -- (304.13,88.17) ;
	\draw [shift={(303.33,86.33)}, rotate = 66.41] [color={rgb, 255:red, 155; green, 155; blue, 155 }  ,draw opacity=1 ][line width=0.75]    (10.93,-3.29) .. controls (6.95,-1.4) and (3.31,-0.3) .. (0,0) .. controls (3.31,0.3) and (6.95,1.4) .. (10.93,3.29)   ;
	%Shape: Arc [id:dp7257242215345958] 
	\draw  [draw opacity=0][dash pattern={on 0.84pt off 2.51pt}] (255.04,134.13) .. controls (258.14,125.26) and (262.67,116.98) .. (268.38,109.6) -- (336.67,162.67) -- cycle ; \draw [color={rgb, 255:red, 74; green, 144; blue, 226 }  ,draw opacity=1 ][dash pattern={on 0.84pt off 2.51pt}] [dash pattern={on 0.84pt off 2.51pt}]  (255.04,134.13) .. controls (258.14,125.26) and (262.67,116.98) .. (268.38,109.6) ;  
	%Curve Lines [id:da04133785683514013] 
	\draw [line width=1.5]    (250.21,163.42) .. controls (248.5,158.42) and (252.5,138.92) .. (255,134.42) ;
	\draw [shift={(255,134.42)}, rotate = 299.05] [color={rgb, 255:red, 0; green, 0; blue, 0 }  ][fill={rgb, 255:red, 0; green, 0; blue, 0 }  ][line width=1.5]      (0, 0) circle [x radius= 2.61, y radius= 2.61]   ;
	\draw [shift={(250.21,163.42)}, rotate = 251.1] [color={rgb, 255:red, 0; green, 0; blue, 0 }  ][fill={rgb, 255:red, 0; green, 0; blue, 0 }  ][line width=1.5]      (0, 0) circle [x radius= 2.61, y radius= 2.61]   ;
	
	% Text Node
	\draw (320.33,166.07) node [anchor=north west][inner sep=0.75pt]    {$O$};
	% Text Node
	\draw (310.33,124.73) node [anchor=north west][inner sep=0.75pt]    {$\rho $};
	% Text Node
	\draw (424.53,165.63) node [anchor=north west][inner sep=0.75pt]    {$z_{0}$};
	% Text Node
	\draw (410.53,92.63) node [anchor=north west][inner sep=0.75pt]    {$z_{1}$};
	% Text Node
	\draw (379.53,68.63) node [anchor=north west][inner sep=0.75pt]    {$z_{2}$};
	% Text Node
	\draw (249.53,84.63) node [anchor=north west][inner sep=0.75pt]    {$z_{3}$};
	% Text Node
	\draw (231.53,116.13) node [anchor=north west][inner sep=0.75pt]    {$z_{4}$};
	% Text Node
	\draw (229.53,162.63) node [anchor=north west][inner sep=0.75pt]    {$z_{5}$};

\end{tikzpicture}}
\caption{The plot of $\Gamma^+$ showing bands and gaps of the fNLS circular condensate, where \smash{$\{z_j=\rho {\rm e}^{{\rm i}\alpha_j}\}_{j=0}^5$} are the endpoints of the bands
respectively and $\rho$ is the radius; $\alpha_0=\arg z_0=0$, $\alpha_5=\arg z_5=\pi$ are fixed. If any of the gaps, shown on the plot, is closed, the genus becomes zero.}
\label{fig:genus one master}
\end{figure}
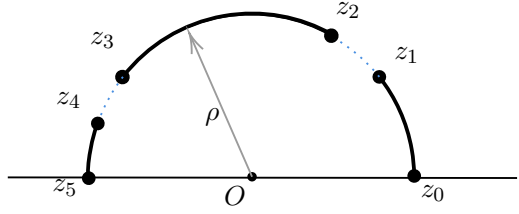

\begin{Corollary}
%\label{g1-master case}
Given the contour $\G^+$ as shown in Figure~$\ref{fig:genus one master}$, the solution to the NDR \eqref{NDR}--\eqref{NDR2} is given by
\begin{align}
&u(z)=u_1(p;a_1,a_2,a_3,a_4)=\frac{-\sqrt{1-p^2}}{\pi}\left(\frac{-p^2+\frac{1}{2}l_1p+A}{R(p)}\right),\label{g1-DOS}\\
&v(z)=v_1(p;a_1,a_2,a_3,a_4)=\frac{\rho\sqrt{1-p^2}}{\pi}\left(\frac{8p^3-4l_1p^2+\bigl(4l_2-l_1^2\bigr)p+B}{R(p)}\right), \label{g1-DOS-v}
\end{align}
where the subindex in $u_1$, $v_1$ indicates the genus, $p=p(z)$ is given by \eqref{conf-map-p},
\begin{gather*}
R(p)=\sqrt{\prod_{j=1}^4(p-a_j)}, \qquad 	l_1=\sum_{j=1}^4a_j, \qquad l_2=\sum_{1\leq i<j\leq 4 }a_ia_j,\\
A=\frac{E(m)}{2K(m)}(a_2-a_4)(a_1-a_3)-\frac{1}{2}(a_1a_2+a_3a_4),\\
B=-\frac{E(m)}{K(m)}(a_2-a_4)(a_1-a_3)l_1+(a_1a_2-a_3a_4)(a_1+a_2-a_3-a_4),\\
m=\frac{(a_1-a_2)(a_3-a_4)}{(a_1-a_3)(a_2-a_4)},
\end{gather*}
and, as in Theorem $\ref{th-main-1}$, the branch of $R(p)$ is chosen so that $R(p)\sim p^2$ as ${p\ra \infty}$.
The effective velocity is then given by
\begin{align}
s(z)=s_1(p;a_1,a_2,a_3,a_4)=-\left(\frac{8p^3-4l_1p^2+\bigl(4l_2-l_1^2\bigr)p+B}{p^2-\frac{1}{2}l_1p-A}\right)\rho.\label{g1-effv}
\end{align}
\end{Corollary}
\begin{proof}
Based on Theorem \ref{th-main-1}, the general solution to the first NDR reads
\begin{align*}
u(z)=u_1(p;a_1,a_2,a_3,a_4)=\frac{1}{\pi}\frac{-\sqrt{1-p^2}\bigl(-p^2+c_1p+c_0\bigr)}{\sqrt{(p-a_1)(p-a_2)(p-a_3)(p-a_4)}},
\end{align*}
where $c_1=\frac{1}{2}l_1$, $c_0=A$ are determined by the (gap-vanishing) normalization conditions:
\begin{align*}
&\int_{a_4}^{a_3}u(p)\frac{{\rm d}p}{\sqrt{1-p^2}}=0,\qquad
\int_{a_2}^{a_1}u(p)\frac{{\rm d}p}{\sqrt{1-p^2}}=0.
\end{align*}
The computation of $v$ is similar and thus omitted. The effective velocity can be computed directly from equation \eqref{v-eff}.
\end{proof}

In the case of $\a_3=\a_4$, one of the gaps disappears and we are in the genus 0 situation, defined only by $\a_1$, $\a_2$.

\begin{Corollary}\label{g0-master case}
In the case $\alpha_4=\alpha_3$ $($genus zero, see Figure~$\ref{fig:genus one master})$, the solution to the NDR \eqref{NDR}--\eqref{NDR2} is given by
\begin{align}
&u(z)=u_0(p;a_1,a_2)=\frac{1}{\pi}\frac{\sqrt{1-p^2}\bigl(p-\frac{1}{2}(a_2+a_1)\bigr)}{\sqrt{(p-a_2)(p-a_1)}},\label{g0-DOS}\\
&v(z)=v_0(p;a_1,a_2)=-\frac{8}{\pi}\frac{\rho\sqrt{1-p^2}\bigl(p^2-\frac{1}{2}(a_2+a_1)p-\frac{1}{8}(a_1-a_2)^2\bigr)}{\sqrt{(p-a_2)(p-a_1)}},\label{g0-DOS-v}
\end{align}
where $p=p(z)$ is given by \eqref{conf-map-p},
$a_j=\cos(\alpha_j)$, $j=1,2$, and the square-root function takes the principal branch.

The effective velocity is then given by
\begin{align}
s(z)=s_0(p;a_1,a_2)=\left(-8p+\frac{(a_2-a_1)^2}{p-(a_2+a_1)/2}\right)\rho.\label{g0-effv}
\end{align}
\end{Corollary}

\begin{proof}
Using the conditions of Theorem \ref{th-main-1} and having in mind that $\hat u(\pm 1)= \hat v(\pm 1)=0$, we obtain \eqref{g0-DOS} and
\begin{align*}
v(z)=-\frac{8\rho}{\pi}\sqrt{\frac{1-p^2}{(p-a_1)(p-a_2)}}\left(p^2-\frac{a_1+a_2}{2}p+c_0\right),
\end{align*}
where $c_0$ is the constant that is determined by the gap vanishing condition
\begin{align*}
\int_{a_2}^{a_1}v(z(p))(1-p^2)^{-1/2}{\rm d}p=0.
\end{align*}
Solving the latter equation, we obtain
\begin{align*}
c_0=-\frac{1}{8}(a_1-a_2)^2.
\end{align*}
Given solutions $u$, $v$, we obtain \eqref{g0-effv} for the effective velocity $s$.
\end{proof}

\begin{Remark}
As it was mentioned in Remark $\ref{rem-break}$, imaginary normalized differentials have a~continuous transition through breaking points, i.e., through points of collapse of a band or a~gap. In particular,
genus zero solutions $($see equations \eqref{g0-DOS}--\eqref{g0-effv}$)$ can be directly obtained from the genus one solutions $($see equations \eqref{g1-DOS}--\eqref{g1-effv}$)$ by taking the limit \smash{$a_3\ra a_4^+$}. That is to say,\looseness=-1
\[u_0(p;a_1,a_2)=\lim_{\substack{a_3\ra a_4^+}}u_1(p;a_1,a_2,a_3,a_4).\]
Similar limits work for the solutions to the density of fluxes and the corresponding effective velocities.
\end{Remark}

\begin{Remark}
As it was shown in the proof of Theorem~$\ref{th-main-1}$, the DOS $u>0$ on $\G^+$; however, the DOF $v$ may have a zero $z_0$ on $\G^+$, which also coincides with the zero of the effective velocity $s$. In Figure~$\ref{fig:genuszeroII}$, we show different cases of the location of $z_0$, defined by
the branch points $(\alpha_1,\alpha_2)$.
\end{Remark}

\begin{figure}[t]
\centering
\includegraphics[width=11cm]{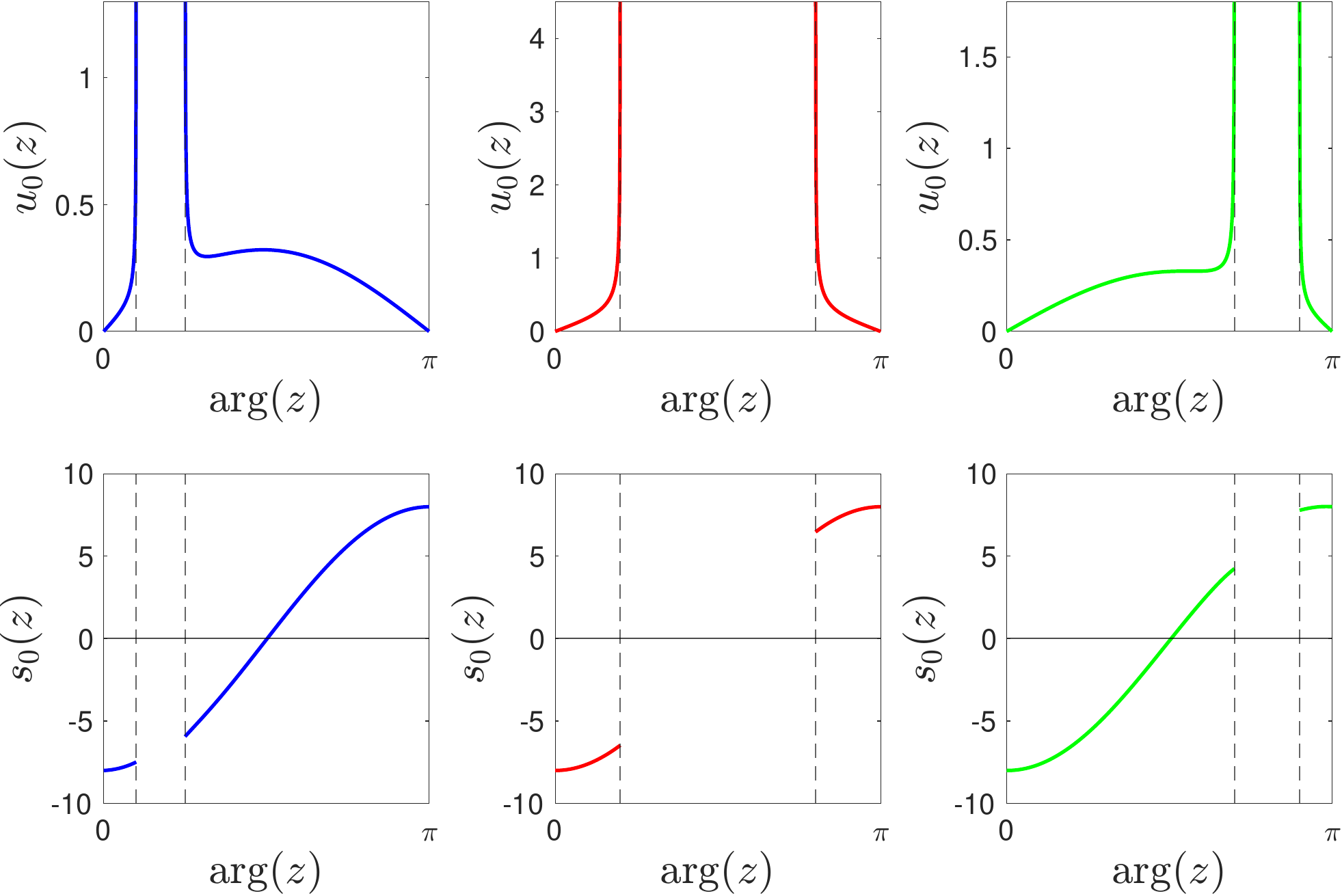}
\caption{Three plots of the DOS $u$ and the corresponding effective velocity based on the results of Corollary~\ref{g0-master case}. Here $\rho=1$. Left: $\alpha_1=0.1\pi$, $\alpha_2=0.25\pi$; the only zero of $s(z)$ lies on the band \smash{$\{\rho {\rm e}^{{\rm i}\theta}\colon \theta\in (\alpha_2,\pi)\}$};
Middle: $\alpha_1=0.2\pi$, $\alpha_2=0.8\pi$; $s(z)$ has no zero on $\G^+$; Right: $\alpha_{1}=0.7\pi$, $\alpha_2=0.9\pi$;
the only zero of $s(z)$ lies in the arc \smash{$\{\rho {\rm e}^{{\rm i}\theta}\colon \theta\in (0,\alpha_1)\}$}.
In each plot, the vertical dashed lines indicate $\alpha_1$ and $\alpha_2$.}
\label{fig:genuszeroII}
\end{figure}

\begin{Remark}\label{g0-cor}
If in the conditions of
Corollary $\ref{g0-master case}$ we further assume $\alpha_2=\pi$, then the solutions to the NDR reduce to
\begin{align}
u(z)&=u_0(p;a_1,-1)=\frac{1}{\pi}\frac{(1- p)\bigl( p+\frac{1-a_1}{2}\bigr)}{\sqrt{( p-a_1)(1- p)}},\label{g0-DOS-2}\\
v( z)&=v_0(p;a_1,-1)=\frac{1}{\pi}\frac{\rho(1- p)\bigl(-8 p^2+4(a_1-1) p+(a_1+1)^2\bigr)}{\sqrt{(1- p)( p-a_1)}}, \label{g0-DOF-2}
\end{align}
%where $p=\cos(\arg(z))=\Re(z)/\rho$,
where $p=p(z)$ is given by \eqref{conf-map-p},
so that
%
%And the effective velocity is given by
\begin{align}
s(z)=\left(-8p+\frac{(a_1+1)^2}{p+(1-a_1)/2}\right)\rho.\label{g0-effv-2}
\end{align}
%It is worth mention that formulas
Expressions
\eqref{g0-DOS-2} and \eqref{g0-effv-2} were calculated in {\rm \cite[equations $(72)$ and $(73)$]{ET2020}.}
\end{Remark}

\section{Modulational dynamics for the circular condensate} \label{sec-Riem-prob}
Solutions to the kinetic equations for the fNLS circular condensate, obtained in Theorems \ref{th-main-1} and~\ref{th-main-2}, look very similar to that for the KdV soliton condensate, obtained in \cite{CERT}. Following the ideas of
\cite{CERT}, in this section we consider the evolution of step function initial conditions for the circular gas modulation equations, i.e., the Riemann problem, that, as in the KdV case, produce rarefaction and dispersive shock wave solutions.

\begin{Remark}
As it was pointed out by one of the referees,
similar phenomena in the context of the ``generalized hydrodynamics'', namely, the shock waves in the zero-entropy GHD, was reported in {\rm \cite{DDKY2017}}.
\end{Remark}

Consider the step function initial data for the modulation equations \eqref{Whit-alp}
\begin{align}
a_1(x,t=0)\equiv \cos(\alpha_1(x,t=0))=\begin{cases}
q_-,& x<0,\\
q_+,& x>0,
\end{cases}\qquad q_+\neq q_-.\label{mod-initial-data}
\end{align}
that corresponds to the genus zero circular condensate described in Remark \ref{g0-cor}, i.e., the initial date
for the branch point $a_1\in\hat\G^+$. That defines (see \eqref{g0-DOS-2}) the DOS:{\samepage
\begin{align}
u(z;x,t=0)=\begin{cases}u_0({p(z)};q_-,-1),& x<0,\\
u_0({p(z)};q_+,-1),& x>0,\end{cases}\label{cond-Riemann-Prob}
\end{align}
and a similar expression for the DOF $v$, see \eqref{g0-DOF-2},
where $q_\pm\in(-1,1)$.}

In what follows, we will consider self-similar solutions to the modulation system \eqref{Whit-alp}, i.e., we consider $x/t=\xi=$ const. It is well known that the behavior of such solutions, including the genus of the corresponding hyperelliptic Riemann surface $\mathfrak{R}=\mathfrak{R}(x,t)$ (i.e., the Riemann surface in the $p$ plane, whose branchcuts are intervals on the real line.) for $u$, $v$
depends on whether $q_->q_+$ or $q_-<q_+$. The first case,
the genus of $ \mathfrak{R}(x,t)$ stays
zero and the dynamics of the DOS $u(z;x,t)$ is characterized by the rarefaction wave solution of
the modulation
equation \eqref{Whit-alp}. The latter case, however, implies immediate wave-breaking, which can be regularized by introducing a genus one surface $ \mathfrak{R}_1(x,t)$ and the corresponding dispersive shock wave DOS $u(z;x,t)$ that connects $u=u_0( p;q_-,-1)$ for large negative and $u=u_0( p;q_+,-1)$ for large positive $x$. Such regularization is well known for describing the dispersive shock wave modulations (see \cite{GP}) of the~KdV equation with step initial data. The two types of behavior of the modulational dynamics of the DOS are shown as in Figure~\ref{fig:density plot DOS}. The following theorem summarizes main results of the section.

\begin{Theorem}\label{Mod-DOS}
Suppose $u(z;x,0)$
%distribution of the DOS
is given by \eqref{cond-Riemann-Prob}
with $a_1(x,0)$ given by \eqref{mod-initial-data}. Then for any $(x,t)\in \R\times \R_+$
we have
\begin{itemize}\itemsep=0pt
\item[$(i)$] If $q_->q_+$,
\begin{align*}
a_1(x,t)=\begin{cases}
q_-,& x<V_{1-}t,\\
-\dfrac{1}{6}\frac{x}{\rho t}+\dfrac{1}{3},& V_{1-}t<x<V_{1+}t,\\
q_+,& x>V_{1+}t,
\end{cases}\qquad \text{and} \qquad V_{1\pm}=-6\rho\left(q_\pm-\frac{1}{3}\right),
\end{align*}
so that the DOS is given by
\begin{align}
u(z;x,t)=u_0(p;a_1(x,t),-1)=\frac{1}{\pi}\frac{(1-p)\bigl(p+\frac{1-a_1(x,t)}{2}\bigr)}{\sqrt{(p-a_1(x,t))(1-p)}},\label{eq: mod-dos}
\end{align}
where $p=p(z)$ is given by \eqref{conf-map-p} and the expression for the DOF is given by
\begin{align}
v(z;x,t)&{}=v_0(p;a_1(x,t),-1)\nonumber\\
&{}=-\frac{8}{\pi}\frac{\rho\sqrt{1-p^2}\bigl(p^2-\frac{1}{2}(a_1(x,t)-1)p-\frac{1}{8}(a_1(x,t)+1)^2\bigr)}{\sqrt{(p+1)(p-a_1(x,t))}}.\label{eq: mod-dof}
\end{align}

\item[$(ii)$] If $q_-<q_+$, then 	$a_2(x,t)$ is uniquely and implicitly determined by the following equation:
\begin{align}
\frac{x}{\rho t}=-2(q_++a_2+q_--1)-\frac{4(a_2-q_-)(q_+-a_2)}{(q_+-q_-)\mu(m) +q_--a_2},\qquad V_{2-}t<x<V_{2+}t,\!\!\!\!\label{def a2 first}
\end{align}
where
\begin{align}
&m=\frac{(1+q_-)(q_+-a_2)}{(q_+-q_-)(1+a_2)},\qquad
\mu(m)=\frac{E(m)}{K(m)}, \label{m,mu}\\
&V_{2-}=\left(\frac{-16a_1^2+8a_1a_3+2a_3^2-8a_1+4a_3+2}{2a_1-a_3+1}\right)\rho,\nonumber\\
&V_{2+}=(-2a_1-4a_3+2)\rho. \label{Def: V2pm}
\end{align}
%the modulated DOS is given by
so that
\begin{align}
u(z;x,t)=\begin{cases}
u_0(p;q_-,-1),& x<V_{2-}t,\\
u_1(p;q_+,a_2(x,t),q_-,-1),& V_{2-}t<x<V_{2+}t,\\
u_0(p;q_+,-1),& x>V_{2+}t,
\end{cases}\label{eq: mod-dos-g1}
\end{align}
where $p=p(z)$ is given by \eqref{conf-map-p}; and the expression for the DOF is given by
\begin{align}
v(z;x,t)=\begin{cases}
v_0(p;q_-,-1),& x<V_{2-}t,\\
v_1(p;q_+,a_2(x,t),q_-,-1),& V_{2-}t<x<V_{2+}t,\\
v_0(p;q_+,-1),& x>V_{2+}t.
\end{cases}\label{eq: mod-dof-g1}
\end{align}

\end{itemize}
\end{Theorem}

\begin{figure}[t]\centering
\subfloat[]{%
\includegraphics[width=0.44\textwidth]{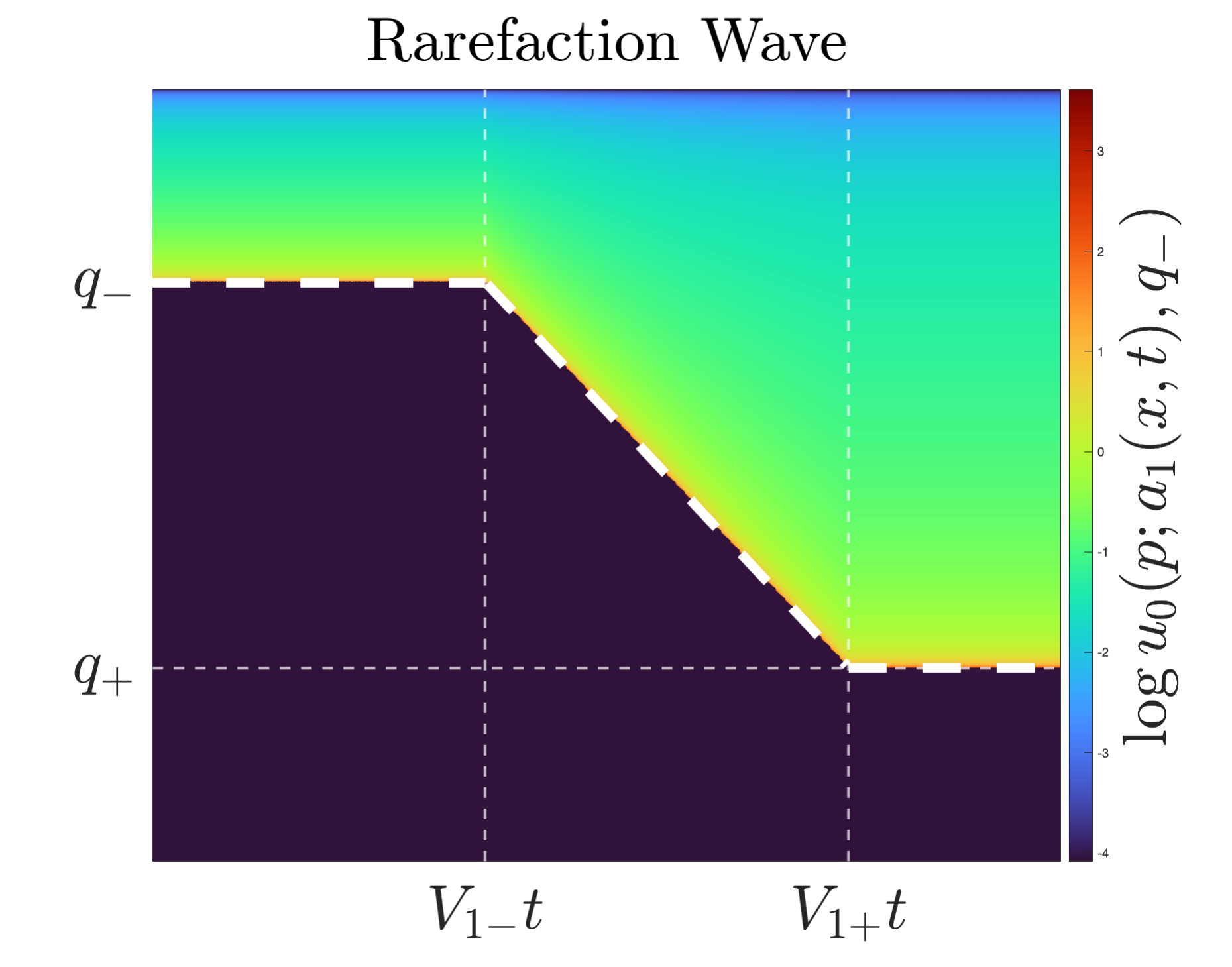}
}
\qquad
\subfloat[]{%
\includegraphics[width=0.44\textwidth]{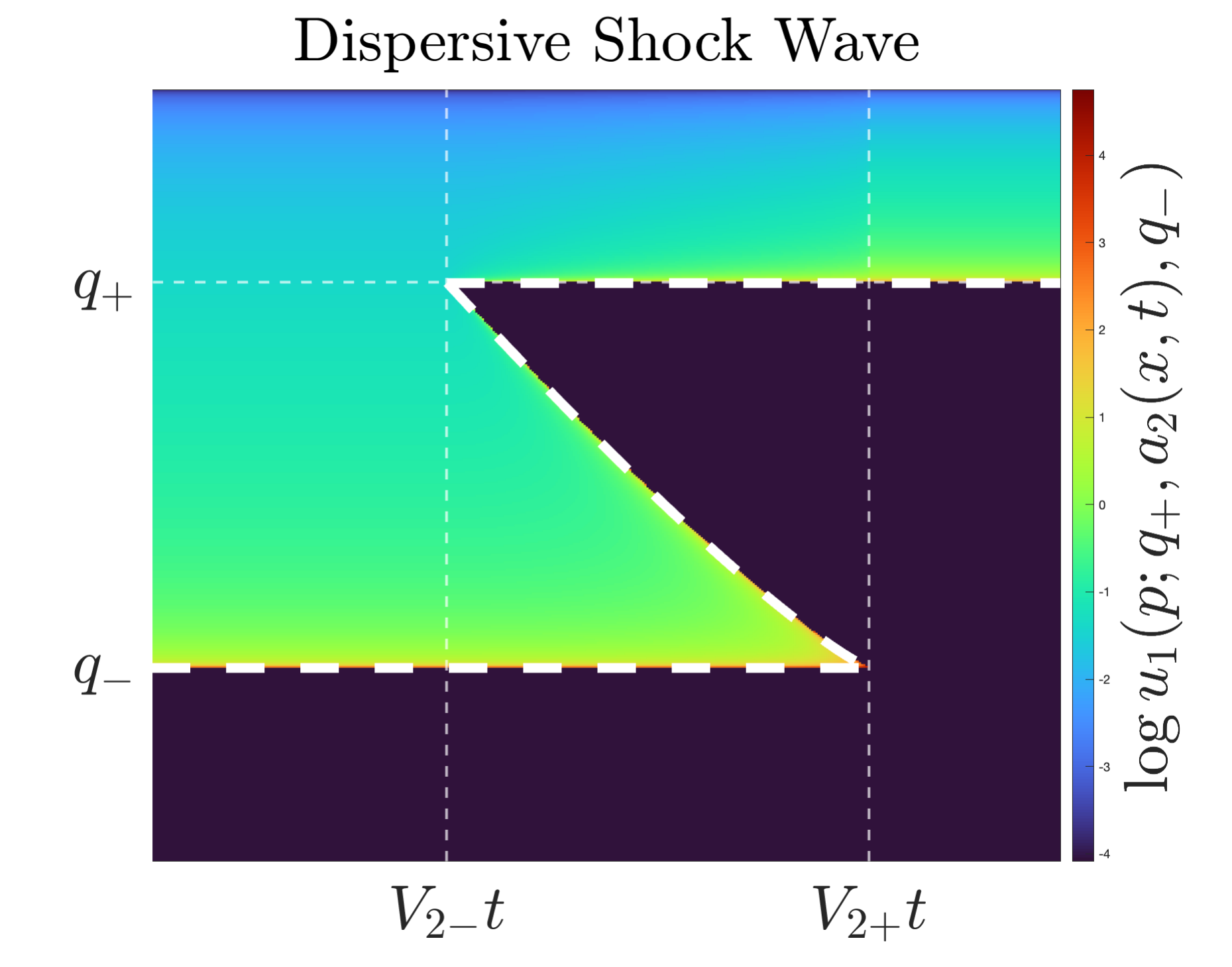}
}
\caption{Solutions to the kinetic equation for the circular condensate with step function initial data (Riemann problem) \eqref{cond-Riemann-Prob}. (a) Rarefaction wave solution when $q_->q_+$; Dashed line: $a_1(x,t)$; Colors: the values of DOS $u_0(p;a_1(x,t),-1)$; (b) Dispersive shock wave solution when $q_-<q_+$. Dashed line: $a_2(x,t)$; Colors: the value of DOS $u_1(p;q_+,a_2(x,t),q_-,-1)$.}
\label{fig:density plot DOS}
\end{figure}

\subsection{Proof of Theorem \ref{Mod-DOS}}

\textbf{Proof of the first part (i).} Applying Theorem \ref{th-main-2}, we obtain the modulation equation for moving the branch point $\alpha_1$:
\begin{align*}
\partial_t\alpha_1+V_1\partial_x\alpha_1=0,
\end{align*}
where
\begin{align*}
V_1 =-6\rho\left(\cos(\alpha_1)-\frac{1}{3}\right).%\label{eq: V1}
\end{align*}
Since $a_1=\cos(\alpha_1)$, the modulation equation is equivalent to
\begin{align*}
\partial_t a_1(x,t)-6\rho\left(a_1-\frac{1}{3}\right)\partial_x a_1(x,t)=0.
\end{align*}
Considering the initial data \eqref{mod-initial-data}, the self-similar solution to the modulation equation is given~by
\begin{align*}
a_1(x,t)=\begin{cases}
q_-,& x<V_{1-}t,\\
-\dfrac{1}{6}\frac{x}{\rho t}+\dfrac{1}{3},& V_{1-}t<x<V_{1+}t,\\
q_+,& x>V_{1+}t,
\end{cases}
\end{align*}
where
\begin{align*}
&V_{1-}=V_1|_{a_1=q_-}=-6\rho\left(q_--\frac{1}{3}\right),\qquad V_{1+}=V_1|_{a_1=q_+}=-6\rho\left(q_+-\frac{1}{3}\right).
\end{align*}
Clearly, $0>V_{1+}>V_{1-}$, which generates a rarefaction wave. In this case, the DOS is given by
\begin{align*}
u(z;x,t)=u_0(p;a_1(x,t),-1)=\frac{1}{\pi}\frac{(1-p)\bigl(p+\frac{1-a_1(x,t)}{2}\bigr)}{\sqrt{(p-a_1(x,t))(1-p)}}.
\end{align*}

\textbf{Proof of the second part (ii):} The previously derived solution (the rarefaction wave) is not well defined for $q_-<q_+$ since a wave breaking occurs immediately. To resolve the issue, it is necessary to introduce higher genus DOS that connects $u_0(p;q_-,-1)$ and $u_0(p;q_+,-1)$. This can be done by using the genus one DOS (see equation \eqref{g1-DOS}):
\begin{align*}
u(z;x,t)=u_1(p;a_1=q_+,a_2(x,t),a_3=q_-,-1).
\end{align*}

Following Theorem \ref{th-main-2}, the motion of $a_2(x,t)$ is governed by the following modulation equation:
\begin{align*}
\partial_t a_2+V_2\partial_x a_2=0,
\end{align*}
where
\begin{align}
V_2 =V_2(a_1,a_2,a_3,-1)=-2\rho(a_1+a_2+a_3-1)-\frac{4\rho(a_2-a_3)(a_1-a_2)}{(a_1-a_3)\mu(m) +a_3-a_2}\label{def-V2}
\end{align}
with
\begin{align*}
m=\frac{(1+a_3)(a_1-a_2)}{(a_1-a_3)(1+a_2)},\qquad
\mu(m)=\frac{E(m)}{K(m)}.
\end{align*}
By a direct computation, we obtain
\begin{align*}
&V_{2-}=\lim_{a_2\ra q_-}V_2=\left(\frac{-16a_1^2+8a_1a_3+2a_3^2-8a_1+4a_3+2}{2a_1-a_3+1}\right)\rho,\\
&V_{2+}=\lim_{a_2\ra 1}V_2=(-2a_1-4a_3+2)\rho.
\end{align*}
Moreover, since \[V_{2+}-V_{2-}=\frac{2\rho(6a_1-a_3+5)(a_1-a_3)}{2a_1-a_3+1},\] it is evident that
\begin{align*}
V_{2+}>V_{2-}.
\end{align*}
Then the solution to the modulation equation for $a_2$ is defined implicitly by the following system:
\begin{align}
V_2(a_1=q_+,a_2,a_3=q_-,-1)=\frac{x}{t},\qquad V_{2-}t<x<V_{2+}t.\label{eq: implic-a2}
\end{align}

Since the solution $a_2(x,t)$ is defined implicitly, we need to prove that the function $V_2$ as a~function of $a_2$ is invertible. In our case, it is equivalent to show that for any fixed $a_1$, $a_3$, the function $V_2(a_1,a_2,a_3,-1)$ is monotonic with respect to $a_2$. So, to complete the proof, we need to prove the following proposition.

\begin{Proposition}\label{mono-V2}
$V_2(a_1,a_2,a_3,-1)$, as defined by equation \eqref{def-V2}, is monotonic as a function of $a_2$ on the interval $(a_3,a_1)$ for any $-1<a_3<a_1<1$ being fixed.
\end{Proposition}

Before we prove the proposition, we need the following lemma.
\begin{Lemma}\label{Lemma-mu-ineq}
Let $\mu(m)=\frac{E(m)}{K(m)}$, then
\begin{align}\label{mu-ineq}
1-m<\mu(m)<1-\frac{m}{2},\qquad m\in (0,1).
\end{align}
\end{Lemma}

\begin{proof}It is well known that (see Byrd--Friedman \cite[formulas~(710.00) and (710.02)]{BFbook})
\begin{align}\label{ell-id}
E'(m)=\frac{E-K}{2m},\qquad K'(m)=\frac{E-(1-m)K}{2m(1-m)},
\end{align}
where prime means differentiation with respect to~$m$. Then
\begin{align*}
E-(1-m)K=2m(1-m)K'(m)>0,
\end{align*}
where we have used the fact that $K'(m)>0$ for $m\in (0,1)$. Thus, the first inequality in \eqref{mu-ineq} is proven.
Using both equations \eqref{ell-id}, we obtain
\begin{align*}
\frac{{\rm d}}{{\rm d}m}\left(E-\left(1-\frac{m}{2}\right)K\right)=-\frac{m}{2}K'(m),
\end{align*}
so that
\begin{align*}
E-\left(1-\frac{m}{2}\right)K<0.
\end{align*}
Thus, $\mu(m)<1-\frac{m}{2}$ for $m\in (0,1)$.
\end{proof}

\begin{proof}[Proof of Proposition \ref{mono-V2}]
We prove the statement by contradiction. Suppose there exists $a_2\in (a_3,a_1)$ such that
\begin{align*}
\partial_{a_2}V_2=0.
\end{align*}
This leads to
\begin{align*}
\mu(m)=\frac{(a_2-a_3)\bigl(a_1^2-a_1a_3-3a_2^2+3a_2a_3+a_1-3a_2+2a_3\bigr)}{2(a_1-a_3)\bigl(2a_1a_2-a_1a_3-3a_2^2+2a_2a_3+a_1-2a_2+a_3\bigr)}.
\end{align*}
Using \eqref{m,mu} to express
$a_2$ in terms of $a_1$, $a_3$ and $m$, we obtain
\begin{align}
\mu(m)={}&\frac{1}{2}+\frac{(a_1+1)(a_3+1)}{2(m(a_1-a_3)+a_3+1)(a_1-a_3)}\nonumber\\
&{}+\frac{\bigl[(a_1+1)^2+(a_3+1)^2\bigr]m-(a_3+1)^2}{2(a_1-a_3)\bigl((a_1-a_3)m^2-2(a_1+1)m+a_3+1\bigr)}.\label{mu-right}
\end{align}

Since $a_1>a_3$ and $m\in (0,1)$, the first two terms has no singularities. The denominator of the third term processes two zeros
\[\frac{(a_1+1)\pm\sqrt{a_1^2-a_1a_3+a_3^2+a_1+a_3+1}}{a_1-a_3},\]
but since $m<1$, we see that the right-hand side of \eqref{mu-right} has a unique simple pole at
\begin{align*}
m_{\rm cr}=\frac{a_1+1-\sqrt{a_1^2-a_1a_3+a_3^2+a_1+a_3+1}}{a_1-a_3}.
\end{align*}
On the one hand, since
\[
(a_1+1)^2-\Bigl(\sqrt{a_1^2-a_1a_3+a_3^2+a_1+a_3+1}\Bigr)^2=(a_3+1)(a_1-a_3)>0, \qquad m_{\rm cr}>0.
\]
 On the other hand, since
 \[
 m_{\rm cr}-1/2=(a_1-a_3)^{-1}\Bigl((a_1+a_3)/2+1-\sqrt{a_1^2-a_1a_3+a_3^2+a_1+a_3+1}\Bigr)
 \]
 and
 \[
 ((a_1+a_3)/2+1)^2-\Bigl(\sqrt{a_1^2-a_1a_3+a_3^2+a_1+a_3+1}\Bigr)^2=-(a_1-a_3)^2<0,
\]
we obtain $m_{\rm cr}<1/2$. Since the singularity $m_{\rm cr}\in (0,1/2)$, we consider two cases: (1) $m\in (0,m_{\rm cr})$; (2) $m\in (m_{\rm cr},1)$. In each case, we will construct a contradiction using Lemma~\ref{Lemma-mu-ineq}.
In the first case, subtract right-hand side of \eqref{mu-right} by $1-m/2$, we get
\begin{align}
\frac{m^2 F(m)}{2(m(a_1-a_3)+a_3+1)((a_1-a_3)m^2-2(a_1+1)m+a_3+1)},\label{pf-case1}
\end{align}
where
\begin{align*}
F(m)={}&(a_1-a_3)^2m^2-(a_1-a_3)(3a_1-2a_3+1)m\\
&{}+a_3^2-(3a_1+1)a_3+3a_1^2+3a_1+1.
\end{align*}
Since $a_1>a_3$, the denominator of \eqref{pf-case1} is positive for any $m\in (0,m_{\rm cr})$. From the expression of the quadratic function $F(m)$, we see that the axis of symmetry is
\[\frac{3a_1-2a_3+1}{2(a_1-a_3)},\]
which is obviously strictly greater than 1. This implies $F(m)$ is a decreasing function for $m\in (0,1)$. Thus, $F(m)\geq F(1)=(a_1+1)^2>0$ and we have shown that
\begin{align*}
\text{right-hand side of \eqref{mu-right}}>1-\frac{m}{2}.
\end{align*}
However, due to Lemma \ref{Lemma-mu-ineq}, this contradicts the inequality satisfied by $\mu(m)$.

In the second case, subtract right-hand side of \eqref{mu-right} by $1-m$, we get
\begin{align}
\frac{m(m-1) G(m)}{2(m(a_1-a_3)+a_3+1)((a_1-a_3)m^2-2(a_1+1)m+a_3+1)},\label{pf-case2}
\end{align}
where \begin{align*}
G(m)=(a_1-a_3)^2m^2-(a_1-a_3)(3a_1-a_3+2)\frac{m}{2}-\frac{(a_3+1)^2}{2}.
\end{align*}
In this case, the denominator is negative since $m>m_{\rm cr}$. As for the quadratic function $G(m)$, it is easy to check $G(0)<0$ and $G(1)<0$, together with the fact that the leading coefficient of $G$ is positive, we have $G(m)<0$ for any $m\in (0,1)$. This implies the expression \eqref{pf-case2} is negative and we have shown
\begin{align*}
\text{right-hand side of \eqref{mu-right}}<1-m.
\end{align*}
Again, due to Lemma \ref{Lemma-mu-ineq}, this contradicts the inequality satisfied by $\mu(m)$.

Hence, we have shown
\[\partial_{a_2}V_2\neq 0,\qquad \forall a_2\in (a_3,a_1).\]
And this means $V_2(a_1,a_2,a_3)$ is a monotonic function of $a_2$ for $a_2\in (a_3,a_1)$.
\end{proof}

Based on Proposition \ref{mono-V2}, we have shown that the solution $a_2$ defined by the equation \eqref{eq: implic-a2} is well defined. Now, let's check the boundary behaviors as $x\ra V_{2+}t$ and $x\ra V_{2-}t$. A direct calculation shows
\begin{align*}
&\lim_{x\ra V_{2-}t}u_1(p;a_1=q_+,a_2(x,t),a_3=q_-,-1)=u_0(p;q_-,-1),\\
&\lim_{x\ra V_{2+}t}u_1(p;a_1=q_+,a_2(x,t),a_3=q_-,-1)=u_0(p;q_+,-1).
\end{align*}

Thus, the genus one DOS, as given by
\begin{align*}
u(z;x,t)=\begin{cases}
u_0(p;q_-,-1),& x<V_{2-}t,\\
u_1(p;q_+,a_2(x,t),q_-,-1),& V_{2-}t<x<V_{2+}t,\\
u_0(p;q_+,-1),& x>V_{2+}t.
\end{cases}%\label{DSW-DOS}
\end{align*}
connects two genus zero DOS: $u_0(p;q_-,-1)$ and $u_0(p;q_+,-1)$. Similarly, using the expressions for the DOF (namely, equations \eqref{g1-DOS-v} and \eqref{g0-DOS-v}), we get the modulated DOF as given by equations \eqref{eq: mod-dof-g1} and \eqref{eq: mod-dof} respectively. This completes the proof for Theorem \ref{Mod-DOS}.

\section{Kurtosis in genus 0 and genus 1 circular condensate}\label{sec-kurt}

In this section, we will compute the fourth normalized moment $\kappa=\bigl\langle|\psi|^4\bigr\rangle/\bigl\langle|\psi|^2\bigr\rangle^2$ of the fNLS circular condensate $|\psi|$-the kurtosis. In the genus 0 case, we obtain that the kurtosis for the condensate is always 2, while in the genus 1 case, the kurtosis is greater than 2 but finite. Below we will give the explicit formulae for computing kurtosis in genus 0 and genus 1 case. The main tool is to use the formulae of computing the averaged conserved quantities for the fNLS soliton gas, which are recently developed by the authors in~\cite{TW}.

It is well known that the fNLS has infinite many conservation laws ($(f_j)_t=(g_j)_x$, $j\geq 1$), where the densities and the currents, $f_j$ and $g_j$, can be determined recursively (see, for example, Wadati's paper \cite{Wadati}). In order to compute the kurtosis for the circular condensate, we will need the first few densities and currents:
\begin{align*}
f_1=|\psi|^2,\qquad
f_3=|\psi|^4+\psi\overline{\psi}_{xx},\qquad
g_2=|\psi_x|^2-|\psi|^4-\psi\overline{\psi}_{xx}.
\end{align*}

According to \cite{TW}, the averaged conserved quantities are given by
\begin{align*}
&\braket{f_1}=2I_1:=4\int_{\G^+}(\Im{z})u(z)|{\rm d}z|,\\
&\braket{f_3}=-\frac{8}{3}I_3:=-\frac{16}{3}\int_{\G^+}\bigl(\Im{z^3}\bigr)u(z)|{\rm d}z|,\\
&\braket{g_2}=-2J_2:=-4\int_{\G^+}\bigl(\Im{z^2}\bigr)v(z)|{\rm d}z|,
\end{align*}
where
\begin{align}
I_j=2\int_{\Gamma_+}u(\xi)\Im \xi^j|{\rm d}\xi|,\label{eq:Ij}\qquad
J_j=2\int_{\Gamma_+}v(\xi)\Im \xi^j|{\rm d}\xi|,
\end{align}
and $\G^+$ is the support of the circular condensates.

Since the total derivatives do not contribute to the average, through integration by parts, the following identities follow
\begin{align*}
\braket{f_3}=\bigl\langle|\psi|^4-|\psi_x|^2\bigr\rangle,\qquad
\braket{g_2}=\bigl\langle 2|\psi_x|^2-|\psi^4|\bigr\rangle.
\end{align*}

Then by the definition of kurtosis, we have
\begin{align}
\kappa=\frac{\bigl\langle|\psi|^4\bigr\rangle}{\bigl\langle|\psi|^2\bigr\rangle^2}=-\frac{\frac{4}{3}I_3+\frac{1}{2}J_2}{I_1^2}.\label{eq:kurtosis def}
\end{align}

Below, we use the formula \eqref{eq:kurtosis def} to derive formulas for computing the kurtosis for the genus~0 and genus 1 condensate. The main ingredient of the computation is computing the averaged densities $\bigl(I_j\bigr)$ and averaged fluxes $\bigl(J_j\bigr)$. The following proposition provides a fairly simple way to compute these quantities.

\begin{Proposition}\label{prop:avg den}
Let $\G^+$ be the contour for the circular condensate associating with the hyperelliptic Riemann surface $\mathfrak{R}_n$, then the averaged densities $\bigl(I_j\bigr)$ and the averaged fluxes $\bigl(J_j\bigr)$ can be computed by the following formulae: for $j\in \mathbb{Z}_+$,
\begin{align}
&I_j=2\pi {\rm i} \rho^{j+1}\Res\left\{\frac{P(z(p))U_{j-1}(p)}{R(z(p))},\,p=\infty\right\},\label{eq:avg den}\\
&J_j=2\pi {\rm i} \rho^{j+1}\Res\left\{\frac{Q(z(p))U_{j-1}(p)}{R(z(p))},\,p=\infty\right\},\label{eq:avg flux}
\end{align}
where $P$, $Q$, $R$ are defined in Theorem~$\ref{th-main-1}$ and
\begin{align}\label{Ujp}
U_j(p)=\frac{\sin[(j+1)\arccos(p)]}{\sin[\arccos(p)]}.
\end{align}
is the $j$-th Chebyshev polynomial of the second kind.
\end{Proposition}
\begin{proof}
Using \eqref{eq:Ij} and \eqref{Ujp} and changes of variables $\xi=\rho {\rm e}^{{\rm i}\theta}$, $p=\cos\theta$, we calculate
\begin{align}
I_j&= 2\int_{\G^+}u(\xi)\Im \xi^j|{\rm d}\xi|
=2\rho^{j+1}\int u\bigl(\rho {\rm e}^{{\rm i} \theta}\bigr)\sin(j\theta) \,{\rm d}\theta
\nonumber\\
& = -2\rho^{j+1}\int_{\hat\G^+}u(z(p))\frac{\sin(j\theta)}{\sin\theta}\,{\rm d}p
=-\rho^{j+1}\oint_{\hat\gamma}\frac{P(z(p))U_{j-1}(p)}{R(z(p))}\,{\rm d}p\label{Cauchy}\\
&= 2\pi {\rm i} \rho^{j+1}\Res\left\{\frac{P(z(p))U_{j-1}(p)}{R(z(p))},\,p=\infty\right\}\label{Residues},
\end{align}
where $U_{j-1}(p)$ is the $(j-1)$-th Chebyshev polynomial of the second kind. The equality \eqref{Cauchy} comes from an application of the Cauchy's theorem on the Riemann surface $\mathfrak{R}_n$, the loop $\hat \gamma$ encloses $\hat\G^+$ counterclockwisely. Then a direct application of the residues theorem leads to the equality \eqref{Residues}, which is exact the formula \eqref{eq:avg den} for computing the averaged densities. By replacing $u$, $P$ by $v$, $Q$ respectively, one get the formula \eqref{eq:avg flux} for computing the averaged fluxes.
\end{proof}

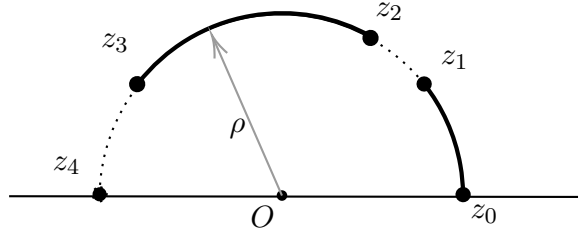
\begin{figure}
    \centering
   \resizebox{0.5\textwidth}{!}{

\tikzset{every picture/.style={line width=0.75pt}} %set default line width to 0.75pt

\begin{tikzpicture}[x=0.75pt,y=0.75pt,yscale=-1,xscale=1]
	
	%Straight Lines [id:da9415440149591279]
	\draw    (207.33,163) -- (477,163.33) ;
	%Shape: Arc [id:dp3332794895139495]
	\draw  [draw opacity=0] (268.03,109.83) .. controls (280.72,93.61) and (299.22,81.84) .. (321.1,77.99) .. controls (341.66,74.38) and (361.8,78.36) .. (378.64,87.91) -- (336.08,163.19) -- cycle ; \draw    (268.03,109.83) .. controls (280.72,93.61) and (299.22,81.84) .. (321.1,77.99) .. controls (341.66,74.38) and (361.8,78.36) .. (378.64,87.91) ; \draw [shift={(378.64,87.91)}, rotate = 25.72] [color={rgb, 255:red, 0; green, 0; blue, 0 }  ][fill={rgb, 255:red, 0; green, 0; blue, 0 }  ][line width=0.75]      (0, 0) circle [x radius= 3.35, y radius= 3.35]   ; \draw [shift={(268.03,109.83)}, rotate = 312.11] [color={rgb, 255:red, 0; green, 0; blue, 0 }  ][fill={rgb, 255:red, 0; green, 0; blue, 0 }  ][line width=0.75]      (0, 0) circle [x radius= 3.35, y radius= 3.35]   ;
	%Shape: Arc [id:dp5363284559358221]
	\draw  [draw opacity=0][line width=1.5]  (404.06,109.75) .. controls (412.6,120.61) and (418.66,133.67) .. (421.22,148.23) .. controls (422.04,152.93) and (422.47,157.61) .. (422.53,162.23) -- (336.08,163.19) -- cycle ; \draw [line width=1.5]    (404.06,109.75) .. controls (412.6,120.61) and (418.66,133.67) .. (421.22,148.23) .. controls (422.04,152.93) and (422.47,157.61) .. (422.53,162.23) ; \draw [shift={(422.53,162.23)}, rotate = 85.36] [color={rgb, 255:red, 0; green, 0; blue, 0 }  ][fill={rgb, 255:red, 0; green, 0; blue, 0 }  ][line width=1.5]      (0, 0) circle [x radius= 2.61, y radius= 2.61]   ; \draw [shift={(404.06,109.75)}, rotate = 55.82] [color={rgb, 255:red, 0; green, 0; blue, 0 }  ][fill={rgb, 255:red, 0; green, 0; blue, 0 }  ][line width=1.5]      (0, 0) circle [x radius= 2.61, y radius= 2.61]   ;
	%Shape: Circle [id:dp38847979876681604]
	\draw  [fill={rgb, 255:red, 0; green, 0; blue, 0 }  ,fill opacity=1 ] (334.67,162.67) .. controls (334.67,161.56) and (335.56,160.67) .. (336.67,160.67) .. controls (337.77,160.67) and (338.67,161.56) .. (338.67,162.67) .. controls (338.67,163.77) and (337.77,164.67) .. (336.67,164.67) .. controls (335.56,164.67) and (334.67,163.77) .. (334.67,162.67) -- cycle ;
	%Shape: Arc [id:dp030040374752331678]
	\draw  [draw opacity=0][dash pattern={on 0.84pt off 2.51pt}] (379.62,88.86) .. controls (388.96,94.22) and (397.27,101.3) .. (404.06,109.75) -- (336.65,163.92) -- cycle ; \draw [color={rgb, 255:red, 0; green, 0; blue, 0 }  ,draw opacity=1 ][dash pattern={on 0.84pt off 2.51pt}] [dash pattern={on 0.84pt off 2.51pt}]  (379.62,88.86) .. controls (388.96,94.22) and (397.27,101.3) .. (404.06,109.75) ;
	%Shape: Arc [id:dp6483332818512029]
	\draw  [draw opacity=0][line width=1.5]  (268.61,109.31) .. controls (281.31,93.08) and (299.81,81.31) .. (321.69,77.47) .. controls (342.25,73.85) and (362.38,77.83) .. (379.23,87.38) -- (336.67,162.67) -- cycle ; \draw [line width=1.5]    (268.61,109.31) .. controls (281.31,93.08) and (299.81,81.31) .. (321.69,77.47) .. controls (342.25,73.85) and (362.38,77.83) .. (379.23,87.38) ; \draw [shift={(379.23,87.38)}, rotate = 25.72] [color={rgb, 255:red, 0; green, 0; blue, 0 }  ][fill={rgb, 255:red, 0; green, 0; blue, 0 }  ][line width=1.5]      (0, 0) circle [x radius= 1.74, y radius= 1.74]   ; \draw [shift={(268.61,109.31)}, rotate = 312.11] [color={rgb, 255:red, 0; green, 0; blue, 0 }  ][fill={rgb, 255:red, 0; green, 0; blue, 0 }  ][line width=1.5]      (0, 0) circle [x radius= 1.74, y radius= 1.74]   ;
	%Straight Lines [id:da44229970288016984]
	\draw [color={rgb, 255:red, 155; green, 155; blue, 155 }  ,draw opacity=1 ]   (336.67,162.67) -- (304.13,88.17) ;
	\draw [shift={(303.33,86.33)}, rotate = 66.41] [color={rgb, 255:red, 155; green, 155; blue, 155 }  ,draw opacity=1 ][line width=0.75]    (10.93,-3.29) .. controls (6.95,-1.4) and (3.31,-0.3) .. (0,0) .. controls (3.31,0.3) and (6.95,1.4) .. (10.93,3.29)   ;
	%Shape: Arc [id:dp7257242215345958]
	\draw  [draw opacity=0][dash pattern={on 0.84pt off 2.51pt}] (250.22,162.08) .. controls (250.54,143.04) and (257.17,125.07) .. (268.38,110.6) -- (336.67,163.67) -- cycle ; \draw [color={rgb, 255:red, 0; green, 0; blue, 0 }  ,draw opacity=1 ][dash pattern={on 0.84pt off 2.51pt}] [dash pattern={on 0.84pt off 2.51pt}]  (250.22,162.08) .. controls (250.54,143.04) and (257.17,125.07) .. (268.38,110.6) ;  \draw [shift={(250.22,162.08)}, rotate = 275.14] [color={rgb, 255:red, 0; green, 0; blue, 0 }  ,draw opacity=1 ][fill={rgb, 255:red, 0; green, 0; blue, 0 }  ,fill opacity=1 ][dash pattern={on 0.84pt off 2.51pt}][line width=0.75]      (0, 0) circle [x radius= 3.35, y radius= 3.35]   ;
	
	% Text Node
	\draw (320.33,166.07) node [anchor=north west][inner sep=0.75pt]    {$O$};
	% Text Node
	\draw (310.33,124.73) node [anchor=north west][inner sep=0.75pt]    {$\rho $};
	% Text Node
	\draw (424.53,165.63) node [anchor=north west][inner sep=0.75pt]    {$z_{0}$};
	% Text Node
	\draw (410.53,92.63) node [anchor=north west][inner sep=0.75pt]    {$z_{1}$};
	% Text Node
	\draw (379.53,69.63) node [anchor=north west][inner sep=0.75pt]    {$z_{2}$};
	% Text Node
	\draw (249.53,84.63) node [anchor=north west][inner sep=0.75pt]    {$z_{3}$};
	% Text Node
	\draw (226.53,142.63) node [anchor=north west][inner sep=0.75pt]    {$z_{4}$};

\end{tikzpicture}

   }
    \caption{The curve $\G^+$ for computing the genus 1 kurtosis and its modulational dynamics. Here $z_4=\rho {\rm e}^{{\rm i}\alpha_4}=-\rho$ and $z_0=\rho {\rm e}^{{\rm i}\alpha_0}=\rho$ are fixed. When considering the modulational dynamics, $\alpha_1=\arg z_1=\arccos(q_+)$, $\alpha_3=\arg z_3=\arccos(q_-)$, where $q_-<q_+$ are given by the initial data for the Riemann problem.  At $t=0$  we have $z_2=z_1$, but for $t>0$  the motion of the point $z_2=\rho\exp\{{\rm i}\arccos(a_2(x,t))\}$ is determined through the equation \eqref{def a2 first}.}
    \label{fig:DSW-spectrum}
\end{figure} 

Based on the above proposition, in order to compute the kurtosis, we just need the first few averaged quantities, namely, $I_1$, $I_3$ and $J_2$. The following proposition gives an explicit formula for computing the kurtosis for the genus one circular condensate.

\begin{Proposition}\label{prop:g1 kur}
Let $\G^+$ be the contour on Figure~$\ref{fig:DSW-spectrum}$, then the kurtosis for the genus one condensate is given by
\begin{align}
\kappa=\frac{\kappa_{\rm num}}{\kappa_{\rm den}}, \label{eq:g1-kurtosis}
\end{align}
where
\begin{gather*}
\kappa_{\rm num}=6a_2^4+24(a_1-a_3+1)a_2^3+4\bigl[(a_1-a_3)^2-(8a_3-10)(a_1-a_3)-3\bigr]a_2^2\\ \qquad{}
+8\bigl[(a_1-a_3)^3+5(a_1-a_3)^2+\bigl(2a_1^2+2a_3^2-7\bigr)(a_1-a_3)-9\bigr]a_2\\ \qquad{}
+2\bigl[3(a_1-a_3)^4+4(a_1-a_3)^3-6(a_1-a_3)^2\bigr]+8\bigl(4a_3+2a_1^2+2a_3^2-9\bigr)(a_1-a_3)\\ \qquad{}
+54-16(a_1-a_3)(a_2+1)\biggl(3\left(a_1+a_2+a_3+\frac{1}{3}\right)^2\!-2(a_1a_2+a_1a_3+a_2a_3)-\frac{28}{3}\biggr)\mu,\\
\kappa_{\rm den}=3\left[(a_1+a_2-a_3+3)(a_1+a_2-a_3-1)-4(a_2+1)(a_1-a_3)\mu\right]^2,\\
\mu=\frac{E(m)}{K(m)},\qquad m=\frac{(1+a_3)(a_1-a_2)}{(a_1-a_3)(a_2+1)}.
\end{gather*}
\end{Proposition}

\begin{proof}
First we use Proposition \ref{prop:avg den} to compute $I_1$, $I_3$ and $J_2$, and then substitute them into the formula \eqref{eq:kurtosis def} to compute the kurtosis. After some algebra, we get the kurtosis formula as stated.
\end{proof}

\begin{Remark}
Notice that the kurtosis for the circular condensate is independent of the radius~$\rho$ of the semicircle $S^+$.
\end{Remark}

\begin{Corollary}\label{thm: g0 kur}
The kurtosis for the genus 0 condensate is always $2$ for any $\alpha_1,\alpha_2\in (0,\pi)$ and $\alpha_1<\alpha_2$.
\end{Corollary}

\begin{proof}
This case can be degenerated from Proposition \ref{prop:g1 kur} by taking the limit $a_3\ra-1+$. In fact, we have
\begin{align*}
&\lim_{a_3\ra-1+}\kappa_{\rm num}=6 (a_{1}-a_{2}+2 )^{2} (a_{1}-a_{2}-2 )^{2},\\
&\lim_{a_3\ra-1+}\kappa_{\rm den}=3 (a_{1}-a_{2}+2 )^{2} (a_{1}-a_{2}-2 )^{2},
\end{align*}
which immediately imply
\begin{align*}
\kappa = 2.
\tag*{\qed}
\end{align*}
\renewcommand{\qed}{}
\end{proof}

From the last section, we know that the genus one DOS/DOF (see equations \eqref{eq: mod-dos-g1} and~\eqref{eq: mod-dof-g1}) actually connects two genus zero DOS/DOF (see equations \eqref{eq: mod-dos} and \eqref{eq: mod-dof}) as $a_2\ra a_1$ or $a_2\ra a_3$. Also, we have already shown that the kurtosis for the genus zero condensate is always~2, it would be interesting to study the dynamic of the kurtosis as the branch point~$a_2$ moving from $a_1$ to $a_3$. A careful analysis to the formula \eqref{eq:g1-kurtosis}, we obtain the following theorem.

\begin{Theorem}\label{Thm:g1 kappa ineq}
Let $a_1$, $a_3$ be fixed and satisfy $-1<a_3<a_1<1$. Then for any $a_2\in [a_3,a_1]$ the genus $1$ kurtosis $\kappa$, given by formula \eqref{eq:g1-kurtosis}, is greater or equal to $2$ and is finite. Moreover, $\kappa=2$ if and only if $a_2=a_1$ or $a_2=a_3$.

\end{Theorem}

\begin{proof}
Using the explicit formula for the genus one kurtosis \eqref{eq:g1-kurtosis}, we define a new function
\begin{align*}
H(\tilde \mu)=\kappa_{\rm num}-2\kappa_{\rm den}
\end{align*}
by replacing $\mu$ with $\tilde \mu$ and consider $\tilde \mu$ as a new variable not depending on $a_1$, $a_2$, $a_3$. To show $\kappa\geq 2$, it suffices to show $H(\tilde \mu)$ is positive for any $\tilde \mu\in (1-m,1-m/2)$, where $m=\frac{(1+a_3)(a_1-a_2)}{(a_1-a_3)(a_2+1)}$. A direct computation shows
\begin{align*}
&H(0)=-32(a_1+1)(a_2+1)(a_1-a_3)(a_2-a_3)<0,\\
&H(1)=-32(a_2+1)(a_3+1)(a_1-a_3)(a_1-a_2)<0,\\
&H(1-m)=32(a_1+1)(a_3+1)(a_2-a_3)(a_1-a_2)>0,\\
&H(1-m/2)=8(a_3+1)^2(a_1-a_2)^2>0.
\end{align*}

Note that the function $H(\tilde \mu)$ is a quadratic function of $\tilde \mu$. The above observation (as visualized on Figure~\ref{fig:Hmu}) shows there is a zero in the interval $(0,1-m)$ and another zero in $(1-m/2,1)$, which implies $H(\tilde \mu)>0$
for any $\tilde \mu \in (1-m,1-m/2)$ and for all $a_2\in (a_3,a_1)$,
see inequality~\eqref{mu-ineq}.
% due to the inequality for $\mu(m)$ as stated in Lemma \ref{mu-ineq}.
This proves that $\kappa>2$ for all $a_2\in (a_3,a_1)$. On the one hand, as $a_2\ra a_1$ or $a_2\ra a_3$, a direct computation shows $\kappa=2$
in both cases.
%the kurtosis will be 2.
On the other hand, if there exists some $\tilde \mu$ such that $H(\tilde \mu)=0$, then either $m=0$ or $m=1$ (otherwise we already show $H(\tilde \mu)>0$). Since $a_1>a_3$, $m=0$ or $m=1$ implies $a_2=a_1$ or $a_2=a_3$ respectively. Thus, we have shown that $\kappa\geq 2$ and $\kappa=2$ if and only if $a_2=a_3$ or $a_1=a_2$.
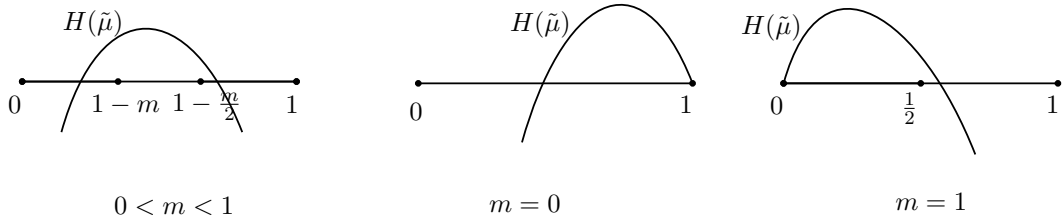
\begin{figure}[ht]
    \centering
   \resizebox{0.9\textwidth}{!}{

\tikzset{every picture/.style={line width=0.75pt}} %set default line width to 0.75pt

\begin{tikzpicture}[x=0.75pt,y=0.75pt,yscale=-1,xscale=1]
	
	%Straight Lines [id:da69345149979432]
	\draw    (81.33,149.33) -- (233,149.33) ;
	\draw [shift={(233,149.33)}, rotate = 0] [color={rgb, 255:red, 0; green, 0; blue, 0 }  ][fill={rgb, 255:red, 0; green, 0; blue, 0 }  ][line width=0.75]      (0, 0) circle [x radius= 1.34, y radius= 1.34]   ;
	\draw [shift={(81.33,149.33)}, rotate = 0] [color={rgb, 255:red, 0; green, 0; blue, 0 }  ][fill={rgb, 255:red, 0; green, 0; blue, 0 }  ][line width=0.75]      (0, 0) circle [x radius= 1.34, y radius= 1.34]   ;
	%Curve Lines [id:da28985090440710826]
	\draw    (103,177.33) .. controls (124.33,101) and (173,101) .. (203,177.33) ;
	%Straight Lines [id:da6597456736782243]
	\draw    (502.33,150.33) -- (654,150.33) ;
	\draw [shift={(654,150.33)}, rotate = 0] [color={rgb, 255:red, 0; green, 0; blue, 0 }  ][fill={rgb, 255:red, 0; green, 0; blue, 0 }  ][line width=0.75]      (0, 0) circle [x radius= 1.34, y radius= 1.34]   ;
	\draw [shift={(502.33,150.33)}, rotate = 0] [color={rgb, 255:red, 0; green, 0; blue, 0 }  ][fill={rgb, 255:red, 0; green, 0; blue, 0 }  ][line width=0.75]      (0, 0) circle [x radius= 1.34, y radius= 1.34]   ;
	%Straight Lines [id:da38687056577991163]
	\draw    (300.33,150.33) -- (452,150.33) ;
	\draw [shift={(452,150.33)}, rotate = 0] [color={rgb, 255:red, 0; green, 0; blue, 0 }  ][fill={rgb, 255:red, 0; green, 0; blue, 0 }  ][line width=0.75]      (0, 0) circle [x radius= 1.34, y radius= 1.34]   ;
	\draw [shift={(300.33,150.33)}, rotate = 0] [color={rgb, 255:red, 0; green, 0; blue, 0 }  ][fill={rgb, 255:red, 0; green, 0; blue, 0 }  ][line width=0.75]      (0, 0) circle [x radius= 1.34, y radius= 1.34]   ;
	%Curve Lines [id:da1926398906920157]
	\draw    (357.67,183) .. controls (379,106.67) and (422,74) .. (452,150.33) ;
	%Curve Lines [id:da5915310304289516]
	\draw    (502.33,150.33) .. controls (523.67,74) and (578.33,113.33) .. (608.33,189.67) ;
	%Straight Lines [id:da14740321212559926]
	\draw    (81.33,149.33) -- (134.33,149.33) ;
	\draw [shift={(134.33,149.33)}, rotate = 0] [color={rgb, 255:red, 0; green, 0; blue, 0 }  ][fill={rgb, 255:red, 0; green, 0; blue, 0 }  ][line width=0.75]      (0, 0) circle [x radius= 1.34, y radius= 1.34]   ;
	\draw [shift={(81.33,149.33)}, rotate = 0] [color={rgb, 255:red, 0; green, 0; blue, 0 }  ][fill={rgb, 255:red, 0; green, 0; blue, 0 }  ][line width=0.75]      (0, 0) circle [x radius= 1.34, y radius= 1.34]   ;
	%Straight Lines [id:da14703401653116788]
	\draw    (180,149.33) -- (233,149.33) ;
	\draw [shift={(233,149.33)}, rotate = 0] [color={rgb, 255:red, 0; green, 0; blue, 0 }  ][fill={rgb, 255:red, 0; green, 0; blue, 0 }  ][line width=0.75]      (0, 0) circle [x radius= 1.34, y radius= 1.34]   ;
	\draw [shift={(180,149.33)}, rotate = 0] [color={rgb, 255:red, 0; green, 0; blue, 0 }  ][fill={rgb, 255:red, 0; green, 0; blue, 0 }  ][line width=0.75]      (0, 0) circle [x radius= 1.34, y radius= 1.34]   ;
	%Straight Lines [id:da227207484499498]
	\draw    (502.33,150.33) -- (578.17,150.33) ;
	\draw [shift={(578.17,150.33)}, rotate = 0] [color={rgb, 255:red, 0; green, 0; blue, 0 }  ][fill={rgb, 255:red, 0; green, 0; blue, 0 }  ][line width=0.75]      (0, 0) circle [x radius= 1.34, y radius= 1.34]   ;
	\draw [shift={(502.33,150.33)}, rotate = 0] [color={rgb, 255:red, 0; green, 0; blue, 0 }  ][fill={rgb, 255:red, 0; green, 0; blue, 0 }  ][line width=0.75]      (0, 0) circle [x radius= 1.34, y radius= 1.34]   ;
	
	% Text Node
	\draw (72,155.73) node [anchor=north west][inner sep=0.75pt]  [font=\normalsize]  {$0$};
	% Text Node
	\draw (118,155.07) node [anchor=north west][inner sep=0.75pt]  [font=\normalsize]  {$1-m$};
	% Text Node
	\draw (162.67,153.73) node [anchor=north west][inner sep=0.75pt]  [font=\normalsize]  {$1-\frac{m}{2}$};
	% Text Node
	\draw (294,159.73) node [anchor=north west][inner sep=0.75pt]  [font=\normalsize]  {$0$};
	% Text Node
	\draw (225.33,155.73) node [anchor=north west][inner sep=0.75pt]  [font=\normalsize]  {$1$};
	% Text Node
	\draw (130.67,211.73) node [anchor=north west][inner sep=0.75pt]  [font=\normalsize]  {$0< m< 1$};
	% Text Node
	\draw (338,210.4) node [anchor=north west][inner sep=0.75pt]  [font=\normalsize]  {$m=0$};
	% Text Node
	\draw (562.67,209.73) node [anchor=north west][inner sep=0.75pt]  [font=\normalsize]  {$m=1$};
	% Text Node
	\draw (443.33,155.73) node [anchor=north west][inner sep=0.75pt]  [font=\normalsize]  {$1$};
	% Text Node
	\draw (646,155.73) node [anchor=north west][inner sep=0.75pt]  [font=\normalsize]  {$1$};
	% Text Node
	\draw (493.33,156.4) node [anchor=north west][inner sep=0.75pt]  [font=\normalsize]  {$0$};
	% Text Node
	\draw (566,153.73) node [anchor=north west][inner sep=0.75pt]  [font=\normalsize]  {$\frac{1}{2}$};
	% Text Node
	\draw (102.33,108.4) node [anchor=north west][inner sep=0.75pt]  [font=\normalsize]  {$H(\tilde{\mu })$};
	% Text Node
	\draw (349.67,109.07) node [anchor=north west][inner sep=0.75pt]  [font=\normalsize]  {$H(\tilde{\mu })$};
	% Text Node
	\draw (477.33,110.4) node [anchor=north west][inner sep=0.75pt]  [font=\normalsize]  {$H(\tilde{\mu })$};

\end{tikzpicture}

   }
    \caption{Three plots of $H(\tilde \mu)$ regarding different values of $m$.}
    \label{fig:Hmu}
\end{figure} 

To show that the kurtosis is finite, it suffices to show that, according to the definition of the kurtosis, $I_1$ is positive and $I_3$, $J_2$ are finite. Since $I_1=\frac{1}{2}\bigl\langle|\psi|^2\bigr\rangle$, it is obvious positive. Since~$u$,~$v$ are integrable with respect to the arc-length measure, we have
\begin{align*}
&|I_3|\leq \bigl\|\Im z^3\bigr\|_{L^\infty(|{\rm d}z|)}\|u(z)\|_{L^1(|{\rm d}z|)}<\infty,\\
&|J_2|\leq \bigl\|\Im z^2\bigr\|_{L^\infty(|{\rm d}z|)}\|v(z)\|_{L^1(|{\rm d}z|)}<\infty.
\end{align*}
Together with definition of the kurtosis, we conclude that $\kappa<\infty$. The proof is done.
\end{proof}

\subsection{The kurtosis for DSW}

In this subsection, we study the modulation dynamics of the kurtosis. As discussed in Section~\ref{sec-Riem-prob},
there are two types of wave phenomenon, the rarefaction wave and the dispersive shock wave. According to Theorem~\ref{thm: g0 kur}, the kurtosis for the rarefaction wave is always $2$. As for the dispersive shock wave, we follow the same setting in Section~\ref{sec-Riem-prob} for the dispersive shock wave. Replacing $a_2$ in equation \eqref{eq:g1-kurtosis} by $a_2(x,t)$ as implicitly defined by equation~\eqref{def a2 first}, we obtain the kurtosis for the dispersive shock wave. Denote the modulated kurtosis by $\kappa_{\rm mod}=\kappa_{\rm mod}(a_2(x,t))$. Recall for each fixed $a_1$, $a_3$ such that $-1<a_3<a_1<1$, we have defined
\begin{align*}
&V_{2-}=\frac{-16a_1^2+8a_1a_3+2a_3^2-8a_1+4a_3+2}{2a_1-a_3+1},\qquad V_{2+}=-2a_1-4a_3+2,
\end{align*}
which in turn define two rays in the $x-t$ plane
\begin{align*}
L_{\pm}:=\{(x,t)\colon x-V_{2\pm}t=0,\,t>0\}.
\end{align*}
These two rays split the $x-t$ plane with $t>0$ into three regions:
\begin{align*}
&D_1:=\{(x,t)\colon x-V_{2+}t>0\},\\
&D_2:=\{(x,t)\colon x-V_{2+}t<0,\, x-V_{2-}t>0\},\\
&D_3:=\{(x,t)\colon x-V_{2-}t<0\}.
\end{align*}
In regions $D_1$ and $D_3$, the kurtosis is $2$ and in the region $D_2$, the kurtosis, according to Theorem~\ref{Thm:g1 kappa ineq}, is strictly great than $2$ and finite as long as $a_1<1$. As an illustrative example, we take $a_1=0.9$, $a_3=-0.4$ and plot the modulated kurtosis $\kappa_{\rm mod}(a_2(x,t))$ in the $x-t$ plane as well as the plot of the kurtosis near the ray $L_-$ in Figure~\ref{fig:Mod Kappa}.

\begin{figure}[t]\centering
\subfloat[]{%
\includegraphics[width=0.48\textwidth]{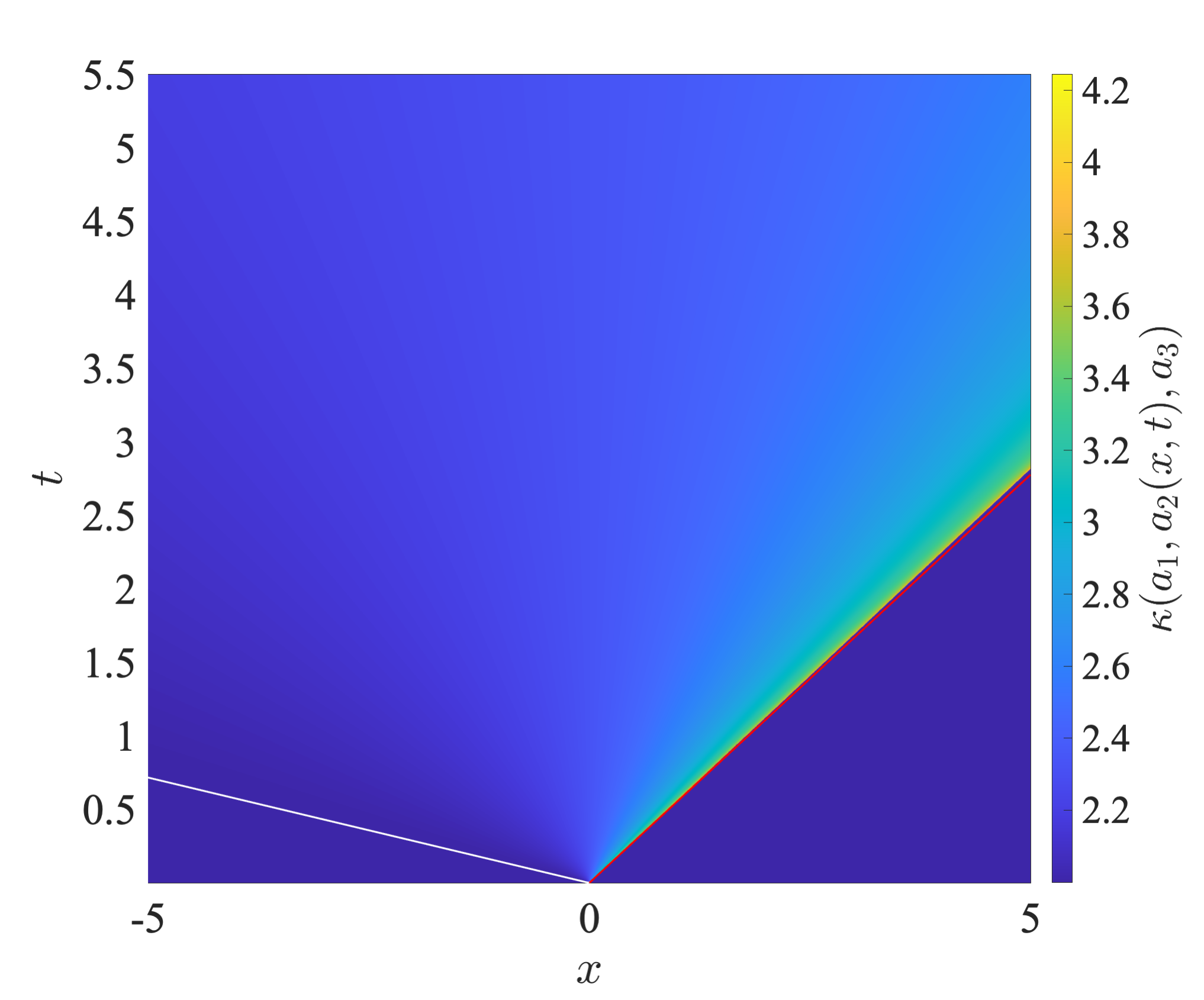}
}
\qquad
\subfloat[]{%
\includegraphics[width=0.45\textwidth]{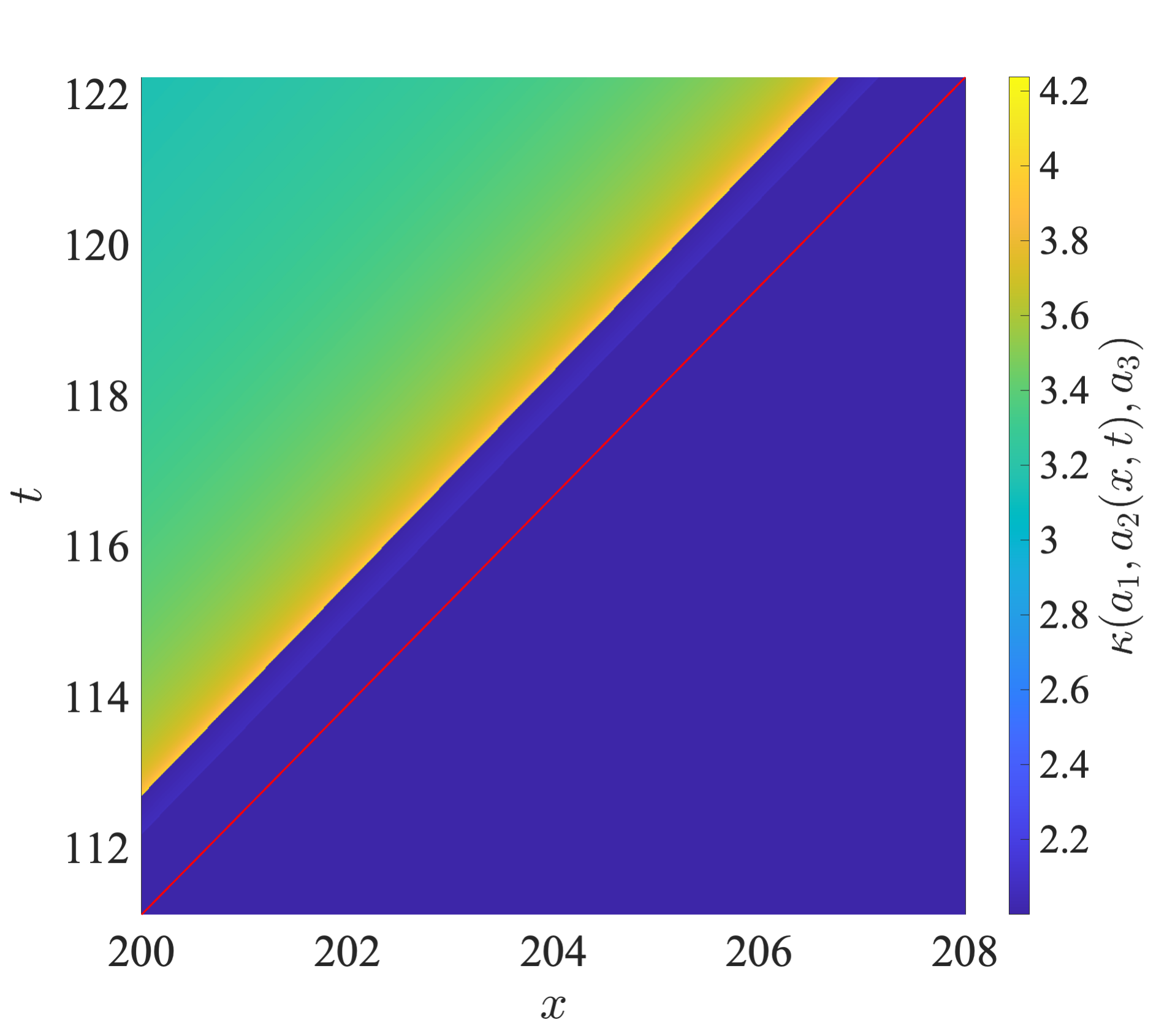}
}
\caption{The density plots of the modulated kurtosis $\kappa(a_1,a_2(x,t),a_3)$, with $a_1=0.9$, $a_3=-0.4$. In~(a), the white line indicates the line $x-V_{2+}t=0$ and the red line indicates the line $x-V_{2-}t=0$, where~$V_{2\pm}$ are defined by equations \eqref{Def: V2pm}. (b) shows the zoom-in density plot near the red line.}\label{fig:Mod Kappa}
\end{figure}

\subsection{Scaling limit of the kurtosis of certain genus one circular condensate}

In this subsection, we consider certain type of limiting configurations of the circular condensate and study the corresponding kurtosis.
Specifically, we set $a_3=-a_2\leq 0$ and study the limit
\begin{align}\label{kurt-lim}
\lim_{(a_1,a_2)\ra (1^-,0^+)}\kappa(a_1,a_2,-a_2)
\end{align}
along a certain path $L$.
In the next theorem, we show that for any given $s>2$, one can always find a path $L_s$ such that the limit in \eqref{kurt-lim} is $s$.

\begin{figure}[t]\centering
\includegraphics[width=0.44\textwidth]{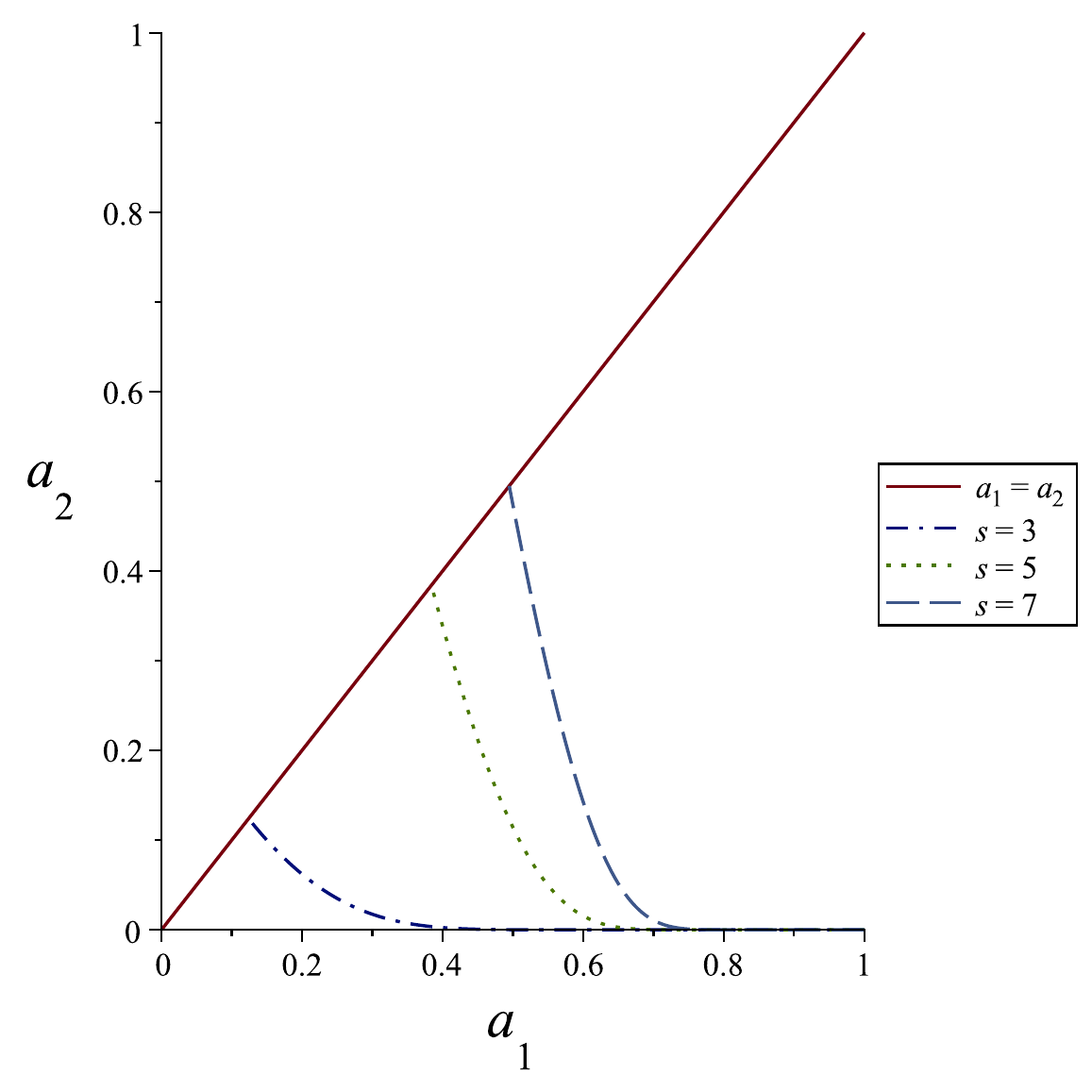}
\caption{Three curves (see the expression in \eqref{curve}) in the $(a_1,a_2)$ plane with $s=3,5,7$. The scaling limit $\lim_{(a_1,a_2)\ra (1,0)}\kappa(a_1,a_2,-a_2)$ along those curves will respectively go to $3$, $5$, $7$. One the line $a_1=a_2$, the kurtosis is equal to 2.}
%\label{fig: sca lim curves}
\end{figure}

\begin{Theorem}\label{thm: abr kur}
For any $s>2$ the limit
\begin{align*}
\lim_{c\ra 0^+}\kappa(a_1,a_2,-a_2)
\end{align*}
taken along the curve
\begin{align}
\begin{cases}
a_1=1-c,\\
a_2=4\exp\left\{\dfrac{8}{3(2-s)c^2}\right\}
\end{cases} \label{curve}
\end{align}
in the parameter plane $(a_1,a_2)$
is equal $s$.
\end{Theorem}

\begin{proof}The kurtosis \eqref{eq:g1-kurtosis} can be written in the following form
\begin{align*}
\kappa(a_1,a_2,-a_2)=\frac{P_1+P_2\mu}{(Q_1+Q_2\mu)^2},
\end{align*}
where $P_1$, $P_2$, $Q_1$, $Q_2$ are all polynomials of $a_1$, $a_2$. Replacing $a_1=1-c$ and we consider the Taylor approximation of those polynomials near $(c,a_2)=(0,0)$, which are given as follows:
\begin{alignat*}{3}
&P_1=-\frac{128}{3}a_2+32c^2+\bigo\bigl(a_2c,a_2^2\bigr),\qquad && P_2=\frac{64}{3}+\bigo(a_2,c),&\\
&Q_1=-4c+8a_2+\bigo\bigl(a_2^2, c^2, a_2c\bigr),\qquad && Q_2=-4+\bigo(a_2,c),&
\end{alignat*}
where $\bigo\bigl(\{A_j\}_{j=1}^n\bigr)$ means the correction term is bounded by some linear combination of $\{A_j\}_{j=1}^n$. Using the asymptotic approximations of complete elliptic integral of the first and the second kind (see formulas~(900.05) and (900.07) in Byrd--Friedman~\cite{BFbook}), it is straightforward to show
\begin{align*}
\mu(m)=\frac{1}{\log\frac{4}{\sqrt{1-m}}}+\bigo((m-1)\log(1-m))\qquad \text{as} \quad m\ra 1^-.
\end{align*}
Notice, as $(c,a_2)\ra (0,0)$, we have
\begin{align*}
m=1-4a_2+\bigo\bigl(a_2^2,a_2c\bigr),
\end{align*}
which implies
\begin{align*}
\mu(m)=\frac{1}{\log\frac{2}{\sqrt{a_2}}}+\bigo\left(\frac{-a_2}{\log(a_2)}\right) \qquad \text{as}\quad (a_2,c)\ra (0,0).
\end{align*}
Then, after some algebraic manipulations, we arrive at the following leading behavior of the kurtosis
\begin{align*}
\kappa = \frac{-\frac{128}{3}a_2+32c^2+\frac{64}{3}\frac{1}{\log\frac{2}{\sqrt{a_2}}}}{\Bigl(4c-8a_2+\frac{4}{\log\frac{2}{\sqrt{a_2}}}\Bigr)^2}+o(1).
\end{align*}

Since $a_2$ is obviously dominated by $-1/\log(a_2)$ as $a_2\ra 0^+$, we can further simplify the kurtosis to
\begin{align*}
\kappa = \frac{32c^2+\frac{64}{3}\frac{1}{\log\frac{2}{\sqrt{a_2}}}}{\Bigl(4c+\frac{4}{\log\frac{2}{\sqrt{a_2}}}\Bigr)^2}+o(1).
\end{align*}
Let us denote $S=\frac{1}{\log\frac{2}{\sqrt{a_2}}}$ and
\begin{align*}
\kappa_{0}=\frac{2c^2+\frac{4}{3}S}{(c+S)^2},
\end{align*}
then $\kappa=\kappa_0+o(1)$ as $c\ra 0+$.
Since
\begin{align*}
\kappa_0=2+\frac{\frac{4}{3}S-4cS-2S^2}{(c+S)^2},
\end{align*}
for $c$, $S$ are sufficiently close to $0$, it is evident that $\kappa_0\geq 2$. Apparently, $S$ will always dominate $cS+S^2$ as $(a_2,c)$ sufficiently close to $(0,0)$, thus, we have
\begin{align*}
\kappa_0=2+\frac{\frac{4}{3}S}{(c+S)^2}+o(1).
\end{align*}
Note that $S=-\frac{3}{4}(2-s)c^2$, we know, as $c\ra 0+$, $c$ dominates $S$, thus the denominator is then dominated by $c^2$. And we eventually get
\begin{align*}
\kappa=2+\frac{4}{3}\frac{S}{c^2}+o(1),
\end{align*}
whose limit as $c\ra 0+$ will be 2. This completes the proof.
\end{proof}

\begin{Remark}
Since the kurtosis is a continuous function of $a_1$, $a_2$, $a_3$ as long as $a_1>a_2>a_3$, Theorem~$\ref{thm: abr kur}$ also implies that for any given number that is greater than $2$, there exists certain configuration $($a genus one circular condensate$)$ such that the kurtosis equals the given number.
\end{Remark}

\subsection*{Acknowledgements}

The authors would also like to thank the anonymous referees for useful suggestions and comments. Fudong Wang is supported by Guangdong Basic and Applied Basic Research Foundation B24030004J. The work of Alexander Tovbis is supported in part by NSF Grant DMS-2009647.

\pdfbookmark[1]{References}{ref}
\LastPageEnding

\end{document}